\documentclass[peerreview]{IEEEtran}
\usepackage{amsmath,amssymb,amsfonts,amsthm}
\usepackage{xfrac}
\usepackage{algorithm}
\usepackage{algpseudocode}
\usepackage{comment}
\usepackage{graphicx}
\usepackage{makecell}
\usepackage{color}
\usepackage[nolist,printonlyused]{acronym}
\usepackage{multirow}
\usepackage[caption=false, font=footnotesize]{subfig}

\usepackage{tikz}
\usepackage{pgf}
\usepackage{pgfplots}
\pgfplotsset{compat=1.18}

\usepackage[sorting=none,style=ieee, minbibnames=1, maxbibnames=4, url=false, doi=false]{biblatex}
\DefineBibliographyStrings{english}{%
  bibliography = {References},
}

\AtBeginBibliography{\footnotesize} 

\newtheorem{theo}{Theorem}
\newtheorem{exam}{Example}
\newtheorem{proposition}{Proposition}
\newtheorem{asum}{Assumption}
	
\newtheorem{lem}{Lemma}	
\newtheorem{corollary}{Corollary}	

\DeclareMathOperator*{\argmin}{arg\,min}
\DeclareMathOperator*{\argmax}{arg\,max}
\DeclareMathOperator*{\subject}{s. t.} 

\usepackage{bbm}

\definecolor{airforceblue}{rgb}{0.36, 0.54, 0.66}

\definecolor{eblue}{HTML}{0082F0}
\definecolor{egreen}{HTML}{0FC373}
\definecolor{epurple}{HTML}{AF78D2}
\definecolor{eyellow}{HTML}{FAD22D}
\definecolor{eorange}{HTML}{FF8C0A}
\definecolor{ered}{HTML}{FF3232}

\definecolor{eblue3}{HTML}{B2D8F9}
\definecolor{egreen3}{HTML}{BDEDD6}
\definecolor{epurple3}{HTML}{E8D6F2}
\definecolor{eyellow3}{HTML}{FEF1C6}
\definecolor{eorange3}{HTML}{FFDBB5}
\definecolor{ered3}{HTML}{FEC3C2}

\begin{document}

\title{Joint Clustering and Prediction of the Quality of Service in Vehicular Cellular Networks \\ 
\thanks{
This work has been submitted to the IEEE for possible publication. Copyright may be transferred without notice, after which this version may no longer be accessible.
}
}

\author{
Oscar Stenhammar$^{\star \dag}$, 
Gábor Fodor$^{\star \dag} $, 
Carlo Fischione$^\star$
\\
\small $^\star$KTH Royal Institute of Technology, Sweden. E-mail: \texttt{ostenh|sraz|gaborf|carlofi@kth.se}
\\
\small $^\dag$Ericsson Research, Sweden. E-mail: \texttt{oscar.stenhammar|gabor.fodor@ericsson.com}\\
}



\begin{acronym}
  \acro{3GPP}{3$^\text{rd}$~Generation Partnership Project}
  \acro{4G}{4$^\text{th}$~Generation}
  \acro{5G}{5$^\text{th}$~generation}
  \acro{6G}{6$^\text{th}$~generation}
  \acro{AI}{artificial intelligence}
  \acro{AC}{admission control}
  \acro{ANN}{artificial neural network}
  \acro{AR}{autoregressive}
  \acro{ARMA}{autoregressive moving average}
  \acro{ARIMA}{auto regressive integrated moving average}
 \acro{AWGN}{additive white Gaussian noise}
  \acro{BPNN}{back propagation neural network}
  \acro{BER}{bit error rate}
  \acro{BCD}{block coordinate descent}
  \acro{BPSK}{binary phase-shift keying}
  \acro{BS}{base station}
  \acro{CDF}{cumulative distribution function}
  \acro{CNN}{convolutional neural network}
  \acro{CSI}{channel state information}
  \acro{CSIR}{channel state information at the receiver}
  \acro{CSIT}{channel state information at the transmitter}
  \acro{CUE}{cellular user equipment}
  \acro{DL}{downlink}
  \acro{DNN}{deep neural network}
  \acro{D-MIMO}{distributed multiple input multiple output}
  \acro{eNB}{eNodeB}
  \acro{FDD}{frequency division duplexing}
  \acro{FL}{federated learning}
  \acro{GRU}{gated recurrent unit}
  \acro{IID}{independent and identically distributed}
  \acro{IoT}{Internet of Things}
  \acro{KDE}{kernel density estimation}
  \acro{KF}{Kalman filter}
  \acro{KPI}{key performance indicator}
  \acro{LOS}{line-of-sight}
  \acro{LS}{least squares}
  \acro{LSTM}{long short-term memory}
  \acro{LTE}{long term evolution}
  \acro{MAC}{medium access control}
  \acro{MAE}{mean absolute error}
  \acro{MAPE}{mean absolute percentage error}
  \acro{mMIMO}{massive multiple input multiple output}
  \acro{MIMO}{multiple input multiple output}
  \acro{ML}{machine learning}
  \acro{MLP}{multilayer perceptron}
  \acro{MMSE}{minimum mean squared error}
  \acro{MSE}{mean squared error}
  \acro{MU-MIMO}{multiuser multiple input multiple output}
  \acro{NLOS}{non-line-of-sight}
  \acro{NMSE}{normalized mean square error}
  \acro{NN}{neural network}
  \acro{NLL}{negative log-likelihood}
  \acro{OFDMA}{orthogonal frequency division multiple access}
  \acro{OFDM}{orthogonal frequency division multiplexing}
  \acro{O-RAN}{open RAN}
  \acro{pQoS}{predictive QoS}
  \acro{PDF}{probability density function}
  \acro{PSD}{positive semidefinite}
  \acro{PRB}{physical resource block}
  \acro{QoS}{quality of service}
  \acro{QoE}{quality of experience}
  \acro{RAN}{radio access network}
  \acro{ReLU}{rectified linear unit}
  \acro{RL}{reinforcement learning}
  \acro{RSCP}{Received Signal Code Power}
  \acro{RSRP}{reference signal received power}
  \acro{RSRQ}{Reference Signal Received Quality}
  \acro{RSSI}{Received Signal Strength Indicator}
  \acro{SLA}{service level agreement}
  \acro{SINR}{signal-to-interference-plus-noise ratio} 
  \acro{SNR}{signal-to-noise ratio}
  \acro{SVD}{singular value decomposition}
  \acro{SVT}{singular value thresholding}
  \acro{TDD}{time division duplexing}
  \acro{TDL}{tapped delay line}
  \acro{UE}{user equipment}
  \acro{UL}{uplink}
  \acro{V2V}{vehicle-to-vehicle}
  \acro{V2X}{vehicle-to-everything}
  \acro{ZF}{zero-forcing}
  \acro{ZMCSCG}{zero mean circularly symmetric complex Gaussian}
\end{acronym}

\maketitle
\pagenumbering{arabic}

\begin{abstract}
    Machine learning models are increasingly deployed in wireless networks with stringent performance requirements. However, dynamic propagation environments and fluctuating traffic densities introduce concept drift, which complicates the ability to maintain accurate predictive machine learning models. We propose a distributed optimization framework that jointly clusters cells and trains cluster-level predictive models, enabling nodes to cooperatively predict quality of service (QoS) distributions under communication constraints. The proposed method models QoS as a multivariate Gaussian/lognormal distribution and uses a novel clustering mechanism that groups cells with similar network conditions, allowing each cell to select the most appropriate predictor without retraining new models for each cell. By leveraging block coordinate descent, our solution efficiently clusters the cells and updates the predictive models to mitigate concept drift, while maintaining a compact model set to minimize computation overhead. Evaluation using data from realistic simulations with the Sionna ray-tracer and the ns-3 simulator shows that the method converges and yields cluster constellations that adapt to changes in the network that cause concept drift. The experimental evaluation focuses on providing a prediction of the distribution latency, jitter, and RSRP over a one-hour prediction horizon.
    The proposed method significantly outperforms the traditional single global predictive model approach and reduces the mean absolute error by 9–27\% compared to local cell-level predictors. This demonstrates that the proposed method effectively captures local variability using far fewer models through scalable distributed clustering.
\end{abstract}

\section{Introduction}
\IEEEPARstart{T}{he} transition toward the \ac{6G} in wireless communication systems is partly driven by emerging applications that demand \ac{QoS} improvements compared to 5G~\cite{10054381, 11126933, 11112590}. 
For example, automated connected users with high mobility impose requirements on low latency, sufficient throughput, and extremely high reliability. In vehicular applications such as teleoperated driving, collision avoidance, or platooning, excessive communication interruptions may directly compromise safety. Network operators therefore offer \acp{SLA} that guarantee predefined \ac{QoS} levels. In this context, predicting future QoS enables proactive and risk-aware scheduling in the \ac{RAN}. By anticipating SLA violations due to congestion or channel variations, the network can allocate resources in advance rather than reacting after a breach has occurred~\cite{9712458, 9779322}. This motivates new architectural and algorithmic solutions across the \ac{RAN}~\cite{10024837}.

Wireless applications with the need for proactive service management, and the rapid advancements of \ac{ML}, have generated great research interest in the development of \ac{pQoS} for next-generation networks~\cite{10155554, 10901462, 11087616}. The \ac{pQoS} algorithms utilize \ac{ML} to predict the \ac{QoS}, which forms the foundation of proactive service management. With the correct setup, \ac{ML} has the ability to make accurate and reliable \ac{QoS} predictions. 
Non-\ac{IID} data, which frequently arise in wireless communication settings, can substantially degrade the prediction accuracy for \ac{ML} algorithms~\cite{palaios2023story}. In vehicular communication networks, dynamic environmental factors such as propagation environments, urban infrastructure, and traffic density lead to concept drift in the received \ac{QoS}~\cite{palaios2023story, 6362634}. 
The term concept drift has various definitions in the literature. In this paper, we adopt the definition of concept drift from~\cite {8496795}, which states that "\emph{Concept drift is a phenomenon in which the statistical properties of a target domain change over time in an arbitrary way}." 

In wireless networks, there can be various types of concept drift, including sudden, gradual, incremental, and recurring drift.
These drifts cause models trained on past data to become misaligned with current conditions. In the context of data-driven \ac{pQoS} ML models, concept drift can significantly deteriorate predictive accuracy, since \ac{ML} relies on training from actual data~\cite{8496795, 11104484}. %
For example, consider a highway road segment where the \ac{pQoS} model was trained under moderate traffic load and stable channel conditions. Over time, traffic density increases due to seasonal commuting patterns, leading to higher interference levels, more frequent handovers, and stronger fluctuations in \ac{SINR}. These spatiotemporal shifts alter the underlying mapping from features (e.g., \ac{RSRP}, channel quality, user speed) to QoS outcomes (e.g., reliability, latency, jitter), meaning that the previously learned model no longer reflects the current network dynamics. This exemplifies concept drift in a \ac{pQoS} context and highlights the need for adaptive learning mechanisms in pQoS systems~\cite{8496795}.


An important aspect of \ac{pQoS} mechanisms is their prediction horizon, the time interval into the future for which the \ac{QoS} is predicted~\cite{10032072}. Short-term predictions, ranging from milliseconds to a few seconds, enable the network to react to rapid channel fluctuations by proactive handover decisions and adaptive modulation and coding. Medium-term prediction horizons, on the order of tens of seconds to minutes, enable operators to anticipate traffic volume peaks and optimize resources across cells. Long-term predictions, ranging from hours to days, can support strategic network management functions and infrastructure planning. By selecting an appropriate prediction horizon, or by combining multiple horizons, \ac{pQoS} can serve both real-time control and broader operational planning. This strengthens \ac{pQoS} as a key enabler of proactive management in future 6G networks. Needless to say, the prediction task becomes more challenging as the prediction horizon increases. 


\subsection{Literature Review}

The majority of prior work on \ac{pQoS} for wireless networks relies on a single global model~\cite{palaios2023story, boban2021predictive}, which is defined as a predictive model trained centrally using aggregated data collected from all devices across a wide geographical region. 
Such a model is typically trained offline using historical data aggregated at a central server to update model weights periodically, rather than in real-time. The model is then deployed at the \acp{BS} to predict future QoS outcomes from locally observed network features. This centralized training paradigm assumes that a single global model can adequately capture the heterogeneous wireless conditions experienced throughout the network and performs well in relatively stable environments where stationarity assumptions approximately hold. However, a large single-model paradigm is inherently vulnerable to concept drift because of spatial and temporal dynamic conditions in smaller regions~\cite{11104484}.

To address these limitations, more adaptive \ac{pQoS} algorithms are necessary to mitigate concept drift and improve robustness. Strategies range from continual learning schemes that update models with fresh data, to domain adaptation approaches that detect and react to distributional shifts~\cite{10032072, 10.1007/s11042-021-11219-x}. These methods aim both to preserve predictive accuracy and to limit the operational cost of retraining in large-scale networks. 

Among the distributed solutions, a clustered version of \ac{FL} has emerged as a promising direction~\cite{Clustered2024Kim, zhang2025clusteredfederatedlearningembedding}. By grouping clients with similar data distributions and performing federated updates within clusters, these approaches reduce negative transfer across heterogeneous devices or locations while retaining the privacy and scalability benefits of \ac{FL}. Clustered federated schemes can dynamically reassign clients as their local distributions evolve, which helps to contain concept drift across user populations~\cite{11045092, s25237370}. However, many of the clustered approaches focus QoS predictions on the user side instead of predicting the QoS on the network side~\cite{taik2022clustered, ghosh2022efficient}, which is where the network management algorithms run. 

Notably, much of the prior literature focuses on short-term QoS forecasting, where the QoS is treated as a single random variable~\cite{palaios2023story}. Considerably less work has addressed longer prediction horizons or the prediction of full QoS distributions. Mid-term prediction efforts exist in previous work. For instance, some studies forecast QoS over horizons on the order of one minute~\cite{8166779}. As the prediction horizon increases, temporal autocorrelation weakens, reducing the usefulness of conventional time-series forecasting models. This motivates the need for probabilistic approaches that explicitly model the Gaussian distribution of future QoS rather than a single expected value~\cite{5450287, mostafavi2025probabilisticdelayforecasting5g}. 

\begin{figure*}[ht]
    \centering
    \includegraphics[width=0.94\linewidth]{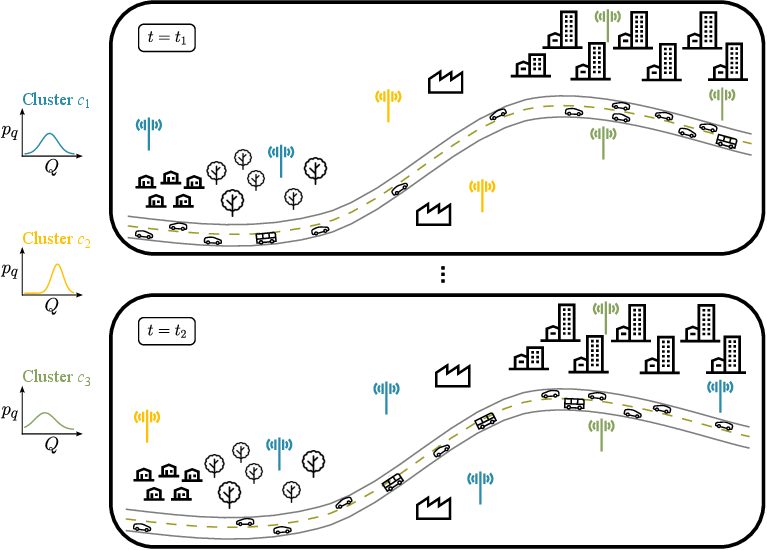}
    \caption{An illustration of the system model. Due to mobile users moving, the scenario evolves between times $t_1$ and $t_2$, resulting in changes in the load and interference in the network. This changes the distribution $p_q$ for a certain \ac{KPI} $Q$ in the different cells. Rather than performing expensive re-training of the existing \ac{pQoS} models, we re-cluster the cells to ensure that each cell is associated with a cluster whose \ac{pQoS} model is well aligned with the scenario of the cell and QoS distribution $p_q$.}
    \label{fig:system}
\end{figure*}

\subsection{Contributions}
This paper takes the next step in \ac{pQoS} by forecasting QoS Gaussian distributions over long-term prediction horizons and by introducing an adaptive, joint clustering and training framework tailored for cellular networks serving \ac{SLA} users. Our main contributions are:

\begin{itemize}
    \item We propose a pQoS framework that groups cells with similar network conditions and learns a dedicated, cluster-specific probabilistic model for each group. The framework predicts the mean vector and covariance matrix of a $d$-dimensional Gaussian distribution over a prediction horizon in the magnitude of hours. 

    \item We formulate this problem as a joint clustering and prediction optimization over cluster assignments, neural network parameters, and the number of clusters. The problem is relaxed via a nuclear-norm regularization. The resulting problem is optimized using \ac{BCD}, alternating between model parameter training with \ac{FL} and cluster assignment updates based on a matrix of distance metrics, \ac{SVT}, and simplex projection.

    \item Under standard assumptions, such as $L$-smoothness and bounded loss, we prove convergence of the proposed BCD procedure to a stationary point. We demonstrate through simulations in heterogeneous deployments that our method converges and outperforms single-model and non-clustered baselines in both distributional accuracy and robustness.
\end{itemize}

Together, these contributions provide a reliable approach for pQoS over long-term prediction horizons for cellular networks that must support \acp{SLA}.

The remainder of the paper is organized as follows: The system model is defined in Section~\ref{sec:system}. The original joint optimization problem to be solved is formulated in Section~\ref{sec:probform}. Section~\ref{sec:probrelax} formulates a relaxed problem that is possible to solve, Section~\ref{sec:bcd} provides the algorithm of the proposed solution, and Section~\ref{sec:analysis} develops the theoretical analysis to prove convergence for the proposed algorithm. Numerical results to assess the proposed method's performance are provided in Section~\ref{sec:eval}. The results are discussed in~\ref{sec:discusion}, and finally, Section~\ref{sec:conclusion} summarizes and concludes the paper.

\section{System Model}
\label{sec:system}
We consider a cellular network deployed over a geographical region comprising a set of $N$ \acp{BS}, each serving individual geographical cells. The system considers user density that varies across the area, for instance, between urban and rural regions, leading to spatially heterogeneous interference levels and network load. The region contains a road grid with diverse propagation environments, along which a set of vehicles travel. Each vehicle runs applications that rely on the cellular network, making reliable QoS predictions essential for ensuring the SLAs. To support these applications, the network operator seeks to manage resources proactively to ensure that SLAs can be met by relying on accurate QoS forecasts over a long-term prediction horizon. 

The goal of our approach is to efficiently capture spatial heterogeneity in the network by partitioning $N$ cells from the network into a smaller number $C\le N$ of clusters. The clusters are formed based on the similarity of the data distributions across cells. Each cluster will be assigned a predictive model. By optimizing the number of clusters, we can balance model complexity and predictive accuracy: too few clusters may fail to capture local variations in QoS, while too many clusters increase training and maintenance overhead. For each cluster, we train a predictive model that outputs the parameters of a multivariate distribution of the QoS metrics. We target the Gaussian distribution instead of an exact value since the QoS over a long-term prediction horizon is influenced by highly variable and stochastic factors, making point predictions difficult to obtain with high confidence. We target prediction horizons in the magnitude of hours. By producing distributional forecasts over such time horizons, the system supports proactive network management and more reliable decision-making. 

The vector \( \mathbf{\hat{y}}_i \) representing the QoS distribution parameters in cell \( i\) can be predicted by
\begin{equation}
\mathbf{\hat{y}}_i (t+h) = \boldsymbol{\omega} (\mathbf{x}_i; \boldsymbol{\theta}_c),
\label{eq:y_hat}
\end{equation}
using the predictive function $\boldsymbol{\omega}$ with weights $\boldsymbol{\theta}_c$, specific for each cluster $c$.
The cell-specific input data \( \mathbf{x}_i \) include network\mbox{-}, mobility\mbox{-}, packet\mbox{-}, and radio-level measurements collected during the most recent sampling period. It also contains the expectation vector and covariance matrix at time $t$ in cell $i$. 
Note that $\mathbf{y}_i (t+h)$ denotes the true distribution parameters at time $t+h$, while $\mathbf{\hat{y}}_i (t+h)$ reflects the predicted parameters.

Observed QoS metrics typically arise from the aggregation of many independent random effects~\cite{ANTONIOU200272}, which, by the central limit theorem, lead to either approximately a Gaussian distribution or a lognormal distribution~\cite{crow1987lognormal}. If a QoS metric follows a lognormal distribution, which is often the case for latency as an example~\cite{ANTONIOU200272}, applying a logarithmic transformation maps the metric to a Gaussian distribution, allowing both cases to be handled within a unified Gaussian modeling framework. This motivates our choice to consider the QoS to follow a Gaussian distribution.

The predictive framework updates associated models with each cell locally at the corresponding \ac{BS} using its own data. The cluster-level models are then updated through an \ac{FL} procedure~\cite{mcmahan2023communicationefficient}, in which \acp{BS} periodically share model gradients with a central server to produce an aggregated update without exchanging raw data. The update period $\tau$ is a design parameter that can be chosen based on system requirements such as network latency and available communication resources. In practice, it is typically defined by the network operator or system designer to balance convergence speed and communication overhead. Shorter periods allow faster adaptation to changing conditions, while longer periods reduce signaling and computational load at the \acp{BS}.

In the literature of QoS and network performance prediction, neural network architectures with three or more hidden layers are frequently used to achieve sufficient model capacity~\cite{palaios2023story}. While such depth improves modeling accuracy, it also substantially increases the number of trainable parameters and the associated training overhead. 
In the predictive function~\eqref{eq:y_hat} of our system model, each cluster \( c \) has a neural network $\boldsymbol{\omega}$ parameterized by the weights \( \boldsymbol{\theta}_c^F \in \mathbb{R}^{M \times F} \), that is shared among all cells in that cluster. To reduce model complexity and training overhead, we restrict the cluster-specific parameters \( \boldsymbol{\theta}_c \in \mathbb{R}^{M \times 1} \) to correspond only to the final layer of the neural network with weights $\boldsymbol{\theta}_c^F = \{\boldsymbol{\theta}_{base}, \boldsymbol{\theta}_c\}$, which are pre-trained on a global dataset. Here, $\boldsymbol{\theta}_{base}\in \mathbb{R}^{M \times F-1}$ represents the first $F-1$ layers. The pre-training dataset is global in the sense that the data is from all the \acp{BS} covering the entire considered geographical area. This design enables transfer learning, whereby knowledge learned from the global predictor is reused and adapted locally with a small number of trainable parameters, allowing efficient specialization to cluster-specific conditions while preserving the generalization capabilities of the global model~\cite{9134370}.

While the full parameter set is given by $\{\boldsymbol{\theta}_{base}, \boldsymbol{\theta}_c\}$, the weights $\boldsymbol{\theta}_{base}$ remain fixed during the cluster-specific updates. Therefore, for the sake of notational brevity, we denote the predictive model with $\boldsymbol{\theta}_c$ as in~\eqref{eq:y_hat}, where the dependence on the constant base parameters $\boldsymbol{\theta}_{base}$ is implicit.

In the following, we provide an illustrative example of how this system model is used.

\begin{exam}
A high-level illustration of the system model is provided in Fig.~\ref{fig:system}. We model the \ac{QoS} $\mathbf{Q}_i(t)\in\mathbb{R}^d$ as a Gaussian random vector with $d$ QoS metrics of interest for each cell \( i\in\mathcal{I} \) at the current time \(t\). To provide a concrete example in this section, $\mathbf{Q}_i(t)$ is a vector consisting of latency $L_i$ and packet loss $P_i$ that follows a bivariate Gaussian distribution:
\[
\mathbf{Q}_i (t) = \begin{bmatrix} L_i \\ P_i \end{bmatrix} \sim \mathcal{N}(\boldsymbol{\mu}_i,\boldsymbol{\Sigma}_i) \quad \forall i=1,...,N.
\]
Here, $L_i\ge0$, $P_i\ge0$. 
Note that we view the elements $L_i(t)$ and $P_i(t)$ of vector $\mathbf{Q}_i(t)$, and the parameters of the Gaussian distribution $\boldsymbol{\mu}_i(t)\in\mathbb{R}^d$ and $\boldsymbol{\Sigma}_i(t)\in\mathbb{R}^{d^2}$, to be time dependent. However, we suppress the notation of time $t$ in the variables for readability. 

The goal is to train the predictive model in~\eqref{eq:y_hat} for each cluster that outputs the Gaussian distribution parameters 
$\boldsymbol{\mu}_{(i, h)}$ and $\boldsymbol{\Sigma}_{(i, h)}$ for the future multivariate distribution $\mathbf{Q}_i (t+h) \sim \mathcal{N}( \boldsymbol{\mu}_{(i, h)}, \boldsymbol{\Sigma}_{(i, h)})$ over a time horizon $h$. By predicting the Gaussian distribution, we can calculate the probability that a user in cell $i$ receives at least the desired future QoS \(\mathbf{Q}_d = \begin{bmatrix} L_{d} & P_{d} \end{bmatrix}^\intercal\), where \(L_{d}\) is the allowed latency and \(P_{d}\) is the allowed packet loss rate. This is done by the following procedure:
\[
\Pr(L_i \leq L_{d}, P_i \leq P_{d}) = \Phi \left( L_{d}, P_{d}; \boldsymbol{\hat{\mu}}_{(i,\, h)}, \boldsymbol{\hat{\Sigma}}_{(i,\, h)} \right), 
\]
where $\Phi$ is a two-dimensional real vector's \ac{CDF} with expectation $\boldsymbol{\hat{\mu}}_{(i,h)}$ and covariance $\boldsymbol{\hat{\Sigma}}_{(i,h)}$. This can be calculated by 
the \ac{PDF} of the predicted bivariate Gaussian distribution.
We assume that $\boldsymbol{\hat{\mu}}{(i,h)}$ and $\boldsymbol{\hat{\Sigma}}{(i,h)}$ are such that the probability $\Pr(L_i<0, P_i<0)$ is negligible. 
This assumption is justified by the fact that the Gaussian approximation is applied in operating regimes where the predicted mean values are strictly positive and the estimation variance is small relative to the mean. Hence, the probability mass assigned to the non-physical negative region is insignificant. 
\end{exam}

Building upon the system model, the next section defines the formal problem formulation addressed in this study.

\section{Problem Formulation}
\label{sec:probform}
To construct $C$ clusters from $N$ cells, optimize the number of clusters $C$, and train a predictive model for each cluster that predicts the set of Gaussian distribution parameters, we formulate an optimization problem in this section.
First, we introduce binary assignment variables \( a_{ic} \):
\[
a_{ic} = 
\begin{cases} 
1, & \text{if cell } i \text{ is assigned to cluster } c, \\ 
0, & \text{otherwise}.
\end{cases}
\]
Here, each cell $i$ must be assigned to exactly one cluster $c$:
\[
\sum_{c=1}^{C} a_{ic} = 1, \quad \forall i=1,\ldots,N.
\]
Depending on latency, scalability, and coordination requirements, this neural network can be deployed on edge servers based on \acp{BS}, centralized RAN controllers, cloud-based network management platforms, or even dedicated accelerator hardware integrated into \ac{O-RAN}.

To formulate this problem, we define three loss terms to optimize in the following three subsections.

\subsection{Prediction Loss}  
\label{sec:pf1}
The prediction loss between the predicted $\mathbf{\hat{y}}_i (t+h)$ from~\eqref{eq:y_hat} and the ground truth \(\mathbf{y}_i(t+h)\) is modeled by:
\[
\sum_{c=1}^{C}\sum_{i=1}^{N} a_{ic}\, \ell\big(\boldsymbol{\omega}(\mathbf{x}_i; \boldsymbol{\theta}_c), \mathbf{y}_i (t+h)\big).
\]
We sum the loss functions $\ell(\cdot, \cdot)$ over the number of clusters $C$ and cells $N$ to retrieve the total loss. For practical purposes, we choose the loss function to be the \ac{MSE}. The binary cluster assignment variable $a_{ic} \in \{0,1\}$ is 1 if cell \(i\) is assigned to cluster \(c\). Otherwise, $a_{ic}=0$. In other words, each cell is counted only once in the sum of the prediction loss.

\subsection{Clustering Consistency Loss}  
To encourage grouping of cells with similar data distributions, we penalize a high separation of data distributions among cells within the same cluster, using a distance measure \( D \):
\[
\lambda\, \sum_{c=1}^{C}\sum_{i,j=1}^{N} a_{ic}\,a_{jc}\, D\big(\mathcal{N}_i, \mathcal{N}_j\big),
\]
where \( \lambda \geq 0 \) is a weighting factor. For simplicity, we denote $\mathcal{N}(\boldsymbol{\mu}_i,\boldsymbol{\Sigma}_i)$ as $\mathcal{N}_i$. If the distributions $\mathcal{N}_i$ and $\mathcal{N}_j$ are dissimilar, the outcome of the distance metric $D\big(\mathcal{N}_i, \mathcal{N}_j\big)$ will be high. On the contrary, the distance metric is low if $\mathcal{N}_i = \mathcal{N}_j$. We will discuss suitable choices of the distance metric in Section~\ref{sec:analysis} and choose a distance metric in Section~\ref{sec:setup}. 

From the distance metric, we want to find the binary weights $a_{ic}$ to cluster cells with similar data distributions.
The weights $a_{ic}$ are utilized to find the cluster sets $\{\mathcal{S}_c \}_{c=1}^C$ from $\mathcal{S}_c = \{ i \in \mathcal{I} \mid c^*_i = \argmax_c a_{ic} \}$ that contains the cells $i$ that belongs to cluster $c$.

\subsection{Cluster Number Penalty}
\label{sec:pf3}
To optimize the number of clusters, we introduce the regularization term $\beta C$, where $\beta > 0$.
Optimizing the number of clusters implies training fewer predictive models. By minimizing the number of predictive models, training and communication become more efficient.

\subsection{Optimization Problem}
Now we are in the position to formulate a structured training problem by adding all three terms defined in the Sections~\ref{sec:pf1}--\ref{sec:pf3} above. We propose the following optimization problem:
\begin{equation}
\begin{aligned}
\min_{C,\; \{a_{ic}\},\; \{\boldsymbol{\theta}_c\}} \quad & \sum_{c=1}^{C}\sum_{i=1}^{N} a_{ic}\, \ell\big(\boldsymbol{\omega}(\mathbf{x}_i; \boldsymbol{\theta}_c), \mathbf{y}_i (t+h)\big)\\
& + \lambda\, \sum_{c=1}^{C}\sum_{i,j=1}^{N} a_{ic} \, a_{jc}\, D\big(\mathcal{N}_i, \mathcal{N}_j\big) + \beta C \\[1mm]
\subject \quad & \sum_{c=1}^{C} a_{ic} = 1, \quad \forall i=1,\ldots,N, \\[1mm]
& a_{ic} \in \{0,1\}, \quad \forall i=1,\ldots,N,\; c=1,\ldots,C.\\[3mm]
\end{aligned}
\label{eq:opt} \\[3mm]
\end{equation}
Because this problem is non-convex, both due to the neural network $\boldsymbol{\omega}$ with non-linear activation functions on weights $\boldsymbol{\theta}_c$ and the combinatorial nature of the clustering assignments, we propose a \ac{BCD} approach to arrive at an approximate solution. One of the BCD blocks is solved at each \ac{BS} separately, and the other BCD block is solved on a central server. In the following section, we present our solution to this problem and show how to perform such a challenging approach.

\section{Solution Approach}
\label{sec:solution}
In this section, we address the computational challenges of the original optimization problem~\eqref{eq:opt} by introducing a suitable relaxation that renders the formulation tractable while preserving its essential structure. Building on a relaxed problem, we then propose a \ac{BCD}–based algorithm that decomposes the optimization into two solvable blocks. Finally, we analyze the convergence properties of the proposed method and establish conditions under which the algorithm converges.

\subsection{Problem Relaxation}
\label{sec:probrelax}
In order to obtain a formulation that allows us to derive an approximate solution, we relax optimization problem~\eqref{eq:opt}. But first, we need to make some definitions.
We set $\boldsymbol{\Theta} = [\boldsymbol{\theta}_1, \ldots, \boldsymbol{\theta}_{C}]  \in \mathbb{R}^{M\times C}$ where $M$ is the length of the vectors $\boldsymbol{\theta}_{c}$ with $c=1,...,C$, containing the weights for the predictive models. 
We define the matrix \(\mathbf{A} \in \mathbb{R}^{N \times C}\) as \(\mathbf{A} = [a_{ic}]_{i=1,\dots,N,\, c=1,\dots,C}\).

With these definitions, we can now relax~\eqref{eq:opt}. We do that by the following two steps:

\begin{itemize}
    \item The elements $a_{ic} \in \{0,1\}$ in $\mathbf{A}$ indicates whether cell $i$ is in cluster $c$ from our definition in Section~\ref{sec:pf1}. Since $a_{ic}$ is binary, it means that if two cells are in the same cluster, their corresponding rows are identical. Intuitively, this means that the rank of $\mathbf{A}$, which counts the number of independent rows, is equal to the number of clusters $C$. For example, if we only have two sets of unique rows in matrix A, the cells will be placed in only 2 different clusters. Such an observation suggests a relaxation of~\ref{eq:opt} by replacing the number of clusters $C$ in the last term of~\eqref{eq:opt} with the rank of $\mathbf{A}$. However, this would make the problem combinatorial and NP-hard~\cite{recht2010guaranteed}, because the rank function is non-convex and discontinuous. To circumvent this difficulty, we replace the number of clusters $C$ in the last term of~\eqref{eq:opt} with the nuclear norm $\|\mathbf{A}\|_*$, defined as the sum of the singular values of $\mathbf{A}$. This convex relaxation is a common relaxation of the rank in the literature, because the number of non-zero singular values is equal to the rank~\cite{candes2008exactmatrixcompletionconvex}.


    \item The original assignment constraints in~\eqref{eq:opt} enforce a hard clustering, requiring $a_{ic} \in \{0,1\}$. This discrete constraint makes the optimization problem combinatorial, non-convex, and generally intractable for large $N$ and $C$. To address this, we relax the binary constraints by allowing the assignment variables $a_{ic}$ to lie in the probability simplex. In this relaxed formulation, $a_{ic}$ can be interpreted as the likelihood that cell $i$ belongs to cluster $c$. The final discrete cluster assignment is obtained by selecting the cluster with the largest probability for each cell $i$ by finding $c^*_i = \argmax_c a_{ic}$. For each cluster, we find the cluster sets $\{\mathcal{S}_c \}_{c=1}^{C_{\text{max}}}$ from $\mathcal{S}_c = \{ i \in \mathcal{I} \mid c^*_i = \argmax_c a_{ic} \}$ that contains the cells $i$ that belongs to cluster $c$. The number of clusters are then found by computing $C = \left| \{ c \in \{1, \dots, C_{\text{max}}\} : \mathcal{S}_c \neq \emptyset \} \right|$.
\end{itemize}

With the relaxations defined above, we now modify~\eqref{eq:opt}. The expression for the relaxed optimization problem becomes:
\begin{equation}
\begin{aligned} 
\min_{\mathbf{A},\, \boldsymbol{\Theta}} \quad & F(\mathbf{A},\boldsymbol{\Theta}) = 
\sum_{c=1}^{C}\sum_{i=1}^{N} a_{ic}\,\ell\big(\boldsymbol{\omega}(\mathbf{x}_i; \boldsymbol{\theta}_c), \mathbf{y}_i(t+h)\big) \\ 
& \quad\quad 
+ \lambda\, \sum_{c=1}^{C}\sum_{i,j=1}^{N}  a_{ic}\,a_{jc}\, D\big(\mathcal{N}_i, \mathcal{N}_j\big) 
+ \beta\|\mathbf{A}\|_*,\\[2mm]
\subject \quad & a_{ic} \geq 0, \:\: \sum_{c=1}^{C} a_{ic} = 1, \quad \forall i=1,\ldots,N, 
\end{aligned}
\label{eq:relaxed_opt}\\[3mm]
\end{equation}

Optimization problem~\eqref{eq:relaxed_opt} is a continuous but still nonconvex optimization problem because the assignments $a_{ic}$ multiply the per-cluster regression losses. Furthermore, computing a global minimizer of the relaxed formulation remains computationally intractable because the formulation can be classified as k-means and similar NP-hard clustering problems~\cite{11045092, Aloise2009}. However,~\eqref{eq:relaxed_opt} is block-multiconvex. The $\boldsymbol{\Theta}$-subproblem is convex for 
fixed $\mathbf{A}$, and the $\mathbf{A}$-subproblem is convex for fixed $\boldsymbol{\theta}$ (when choosing $D$ properly, which we will show in Section~\ref{sec:analysis}). This justifies the use of \ac{BCD}, as we show next.

In order to map optimization problem~\eqref{eq:relaxed_opt} to the \ac{BCD} framework in~\cite{xu2013block}, we first define the functions
\begin{equation}
f(\mathbf{A},\boldsymbol{\Theta}) = \sum_{c=1}^{C}\sum_{i=1}^{N} a_{ic}\,\ell\Big(\boldsymbol{\omega}(\mathbf{x}_i; \boldsymbol{\theta}_c), \mathbf{y}_i(t+h)\Big) + g(\mathbf{A}),
\label{eq:f}
\end{equation}
\begin{equation}
g(\mathbf{A}) = \lambda\, \sum_{c=1}^{C}\sum_{i,j=1}^{N}a_{ic}\,a_{jc}\, D\big(\mathcal{N}_i, \mathcal{N}_j\big),
\label{eq:g}
\end{equation}
\begin{equation}
r_\mathbf{A}(\mathbf{A}) = \beta\|\mathbf{A}\|_*.
\label{eq:r}
\end{equation}
Again, for clarification, $\boldsymbol{\Theta} \in \mathbb{R}^{M\times C}$ and \(\mathbf{A} \in \mathbb{R}^{N \times C}\). 
Now, using vector notation, we can find an alternative way to write $F(\mathbf{A},\boldsymbol{\Theta})$ in~\eqref{eq:relaxed_opt}. 
Let $\mathbf{D}\in\mathbb{R}^{N\times N}$ with $[\mathbf{D}]_{ij}=D(\mathcal{N}_i,\mathcal{N}_j)$, and define vectors $\boldsymbol{\ell}_c,\boldsymbol{a}_c\in\mathbb{R}^N$ by
$[\boldsymbol{\ell}_c]_i=\ell(\boldsymbol{\omega}(\mathbf{x}_i;\boldsymbol{\theta}_c),\mathbf{y}_i(t+h))$ and $[\boldsymbol{a}_c]_i=a_{ic}$.
Using $\mathbf{D}$, $\boldsymbol{\ell}_c$, and $\boldsymbol{a}_c$, we can simplify $F(\mathbf{A},\boldsymbol{\Theta})$ in~\eqref{eq:relaxed_opt} as
\begin{equation}
F(\mathbf{A},\boldsymbol{\Theta}) = \sum_{c=1}^{C} \boldsymbol{a}^\top_{c}\,\boldsymbol{\ell}_c + \lambda\, \sum_{c=1}^{C} \boldsymbol{a}_{c}^\top \mathbf{D}\, \boldsymbol{a}_{c}+\beta\|\mathbf{A}\|_*.
\label{eq:quadratic}
\end{equation}

By the definitions~\eqref{eq:f}-\eqref{eq:quadratic}, we can implement the non-trivial~\ac{BCD} formulation from~\cite{xu2013block}, as presented in the following subsection. 

\subsection{Implementation of Block Coordinate Descent}
\label{sec:bcd}
To approach the objective function $F(\mathbf{A},\boldsymbol{\Theta})$ in~\eqref{eq:quadratic}, we use \ac{BCD}. In the following subsections, we define the two blocks to jointly optimize $\mathbf{A}$ and $\boldsymbol{\Theta}$ in~\eqref{eq:relaxed_opt}.

\subsubsection{Block 1}
We define the problem of block 1 as 
\begin{equation}
\boldsymbol{\Theta}^k = \argmin_{\boldsymbol{\Theta}} \left\{ f(\mathbf{A}^{k-1}, \boldsymbol{\Theta}) + \alpha_{\boldsymbol{\Theta}} \| \boldsymbol{\Theta} - \boldsymbol{\Theta}^{k-1} \|^2 \right\}.
\label{eq:block1}
\end{equation}
Here, $2\alpha_{\boldsymbol{\Theta}} = L_{\boldsymbol{\Theta}}^{k-1}$ denotes the step size for each update of $\boldsymbol{\Theta}^{k-1}$. 
By solving~\eqref{eq:block1}, we update the weights $\boldsymbol{\Theta}^{k}$.

We now calculate the gradient for each predictive model with weights $\boldsymbol{\theta}_c \in \boldsymbol{\Theta}$.
The first-order optimality condition leads to
\[
     \nabla_{\boldsymbol{\theta}_c} f(\mathbf{A}^{k-1}, \boldsymbol{\theta}_c) + L_{\boldsymbol{\theta}_c}^{k-1} (\boldsymbol{\theta}_c - \boldsymbol{\theta}_c^{k-1}) = 0.
\]
By solving for $\boldsymbol{\theta}_c^k$, we get 
\begin{equation}
    \boldsymbol{\theta}_c^k = \boldsymbol{\theta}_c^{k-1} - \frac{1}{L_{\boldsymbol{\theta}_c}^{k-1}} \nabla_{\boldsymbol{\theta}_c} f(\mathbf{A}^{k-1}, \boldsymbol{\theta}_c^{k-1}).
    \label{eq:nonfl}
\end{equation}
This corresponds to a gradient descent step with a step size of \( 1/L_{\boldsymbol{\theta}_c}^{k-1} \). 
The update in~\eqref{eq:nonfl} is performed locally at each \ac{BS} to obtain the model weights $\boldsymbol{\theta}_{c,i}^k$, using only the data associated with cell $i$. To leverage information across the cells in a cluster, a cluster model is aggregated centrally via \ac{FL}, where the updated parameters from all participating cells in a cluster are combined to form the new cluster model parameters
\begin{equation}
\boldsymbol{\theta}_c^k = \frac{1}{| \mathcal{S}_c |} \, \sum_{i \in \mathcal{S}_c} 
\boldsymbol{\theta}_{c,i}^k.
\end{equation}
where $\mathcal{S}_c$ is the set of cells contributing to cluster $c$.



\subsubsection{Block 2}
\label{sec:b2}
Initially, we set $C=C_{\max}=N$. Furthermore, we define the nuclear norm as the regularization function \(r_\mathbf{A}(\mathbf{A})=\|\mathbf{A}\|_*\).
We aim to optimize the following expression:
\begin{equation}
\mathbf{A}^{k} = \argmin_{\mathbf{A}} \left\{ f(\mathbf{A}, \boldsymbol{\Theta}^{k}) + \alpha_{\mathbf{A}} \|\mathbf{A} - \mathbf{A}^{k-1} \|^2 + \|\mathbf{A}\|_* \right\},
\label{eq:block2}
\end{equation}
As in block 1, we define $2\alpha_{\mathbf{A}}=L_\mathbf{A}^{k-1}$ as the step size. 
The first and third terms in~\eqref{eq:block2} represent solving the relaxed optimization problem~\eqref{eq:relaxed_opt}, while the second term controls how fast we update $\mathbf{A}$.

To solve this, we compute the gradient
\begin{equation*}
\begin{aligned}
\nabla_{a_{ic}} f(\mathbf{A}^{k-1}, \boldsymbol{\Theta}^{k}) & = \ell(\boldsymbol{\omega}(\mathbf{x}_i; \boldsymbol{\theta}^{k}_c), \mathbf{y}_i(t+h)) \\
& + 2\lambda \sum_{j=1}^{N} a^{k-1}_{jc}  D\big(\mathcal{N}_i, \mathcal{N}_j\big),
\end{aligned}
\end{equation*}
and perform a gradient step descent:
\[
z_{ic}^{k} = a_{ic}^{k-1} - \frac{1}{L_{\mathbf A}^k} \nabla_{a_{ic}} f(\mathbf{A}^{k-1}, \boldsymbol{\theta}_c^{k}).
\]
The values $z_{ic}^{k}$ are computed locally at each \ac{BS} $i$. They are then transmitted to the central server to form the temporary matrix $\mathbf{Z}^{k}$ as
\[\mathbf{Z}^{k} = \left[ z_{ic}^{k} \right]_{i=1,\dots,N, c=1,\dots,C}.\]

After the gradient step, we apply singular value thresholding (SVT)~\cite{cai2008singularvaluethresholdingalgorithm} to optimize the nuclear norm of $\mathbf{A}$. We update the BCD step as
\[
\mathbf{A}^{k} = \mathcal{S}_{\tau}\left( \mathbf{Z}^{k} \right).
\]
The function \(\mathcal{S}_\tau(\mathbf{X})\) is the SVT operator:
\[
\mathcal{S}_\tau(\mathbf{X}) = \mathcal{S}_\tau(\mathbf{U} \mathbf{\Sigma} \mathbf{V}^T) = \mathbf{U} \mathbf{\Sigma}_\tau \mathbf{V}^T.
\]
Here, \(\mathbf{X} = \mathbf{U} \mathbf{\Sigma} \mathbf{V}^T\) is the singular value decomposition of \(\mathbf{X}\), and \(\mathcal{S}_\tau\) applies soft-thresholding to the singular values $\sigma$, which are located on the diagonal of $\mathbf{\Sigma}$ :
\[
\sigma_\tau = \max(\sigma - \tau, 0).
\]

After we have retrieved $\mathbf{A}^{k}$, we employ Dykstra’s projection algorithm to map each row of $\mathbf{A}^{k}$ onto the probability simplex~\cite{dykstra1983algorithm}, thereby maintaining the probabilistic validity of the cluster assignment. The algorithm iteratively projects the unconstrained updates onto 
the valid range in~\eqref{eq:relaxed_opt}, ensuring that $a_{ic}$ stays between 0 and 1 and the rows of $\mathbf{A}_k$ sum to unity. This projection effectively finds the closest valid $\mathbf{A}_k$ to the BCD updates in the Euclidean sense. 

\begin{algorithm}[t]
\caption{\ac{BCD} for joint clustering and prediction.}
\label{alg:BCD}
\begin{algorithmic}[1]
\State \textbf{Input:} \(\lambda\), \(\tau\), \( \boldsymbol{\Theta}^{0} \), \( \mathbf{A}^{0} \), step sizes \( L_{\boldsymbol \Theta}^k, L_{\mathbf A}^k, \{\mathcal{S}_c \}_{c=1}^C\), 
maximum iterations \( K \).
\For{\( k = 1,2,\dots, K \)}
    \State \textbf{Block 1: Update Neural Network Weights \( \boldsymbol{\Theta} \)}
    \For{\( c = 1, \dots, C \) at cell $i$ if $i \in \mathcal{S}_c$} 
        \State Compute gradient at \ac{BS} $i$:
        \[
        \nabla_{\boldsymbol{\theta}_c} f_i(\mathbf{A}^{k-1}, \boldsymbol{\theta}_{c}^{k-1}).
        \]
        \State Update $\boldsymbol{\theta}_{c,i}^k$ at each \ac{BS} $i$:
        \[
        \boldsymbol{\theta}_{c,i}^k = \boldsymbol{\theta}_{c}^{k-1} - \frac{1}{L_{\boldsymbol \Theta}^k} \nabla_{\boldsymbol{\theta}_{c}} f_i(\mathbf{A}^{k-1}, \boldsymbol{\theta}_{c}^{k-1}).
        \]
    \EndFor

    \State Transmit all $\boldsymbol{\theta}_{c,i}^k$ to the central server.
    
    \State Aggregate the distributed model updates: 
    \[\boldsymbol{\theta}_c^k = \frac{1}{| \mathcal{S}_c |} \, \sum_{i \in \mathcal{S}_c} \boldsymbol{\theta}_{c,i}^k.\]

    \State \textbf{Block 2: Update Cluster Assignment Matrix \( \mathbf{A} \)}
    \State Compute gradient step at each \ac{BS} $i$:
    \begin{equation*}
    \begin{aligned}
    \nabla_{a_{ic}} f(\mathbf{A}^{k-1}, \boldsymbol{\Theta}^{k}) & = \ell(\boldsymbol{\omega}(\mathbf{x}_i; \boldsymbol{\theta}^{k}_c), \mathbf{y}_i(t+h)) \\ & + 2\lambda \sum_{j=1}^{N} a^{k-1}_{jc}  D\big(\mathcal{N}_i, \mathcal{N}_j\big).
    \end{aligned}
    \end{equation*}
    
    \State Update $z_{ic}^{k}$ locally at \acp{BS} $i$, $\forall c$:
    \[
    z_{ic}^{k} = a_{ic}^{k-1} - \frac{1}{L_{\mathbf A}^k} \nabla_{a_{ic}} f(\mathbf{A}^{k-1}, \boldsymbol{\theta}_c^{k}).
    \]
    \State Transmit $z_{ic}^{k}$ from \acp{BS} $i$ to the central server and form
     \[\mathbf{Z}^{k} = \left[ z_{ic}^{k} \right]_{i=1,\dots,N, c=1,\dots,C}.\]
    
    \State Compute \( \mathbf{Z}^{k} = \mathbf{U} \mathbf{\Sigma} \mathbf{V}^T \).
    \State Apply thresholding to singular values:
    \[
    \mathbf{\Sigma}_\tau = \max(\mathbf{\Sigma} - \tau, 0).
    \]
    \State Reconstruct 
    \(
    \mathbf{A}^{k} = \mathbf{U} \mathbf{\Sigma}_\tau \mathbf{V}^T.
    \)
    \State Project \( \mathbf{A}^{k} \) onto probability simplex by Dykstra's~\cite{dykstra1983algorithm}.
    \State Compute the new cluster set $\mathcal{S}_c$ for each cluster $c$:
    \[
    \mathcal{S}_c = \{ i \in \mathcal{I} \mid c^*_i = \argmax_c a_{ic} \}.
    \]

\EndFor
\State \textbf{Output:} \( \boldsymbol{\Theta}^{K}, \mathbf{A}^{K}, \{ \mathcal{S}_c \}_{c=1}^C \).
\end{algorithmic}
\end{algorithm}

Finally, the cluster assignments in $\mathbf{A}_k$ are used to set $c^*_i=\arg\max_c a_{ic}$. We find the cluster sets $\{\mathcal{S}_c \}_{c=1}^{C_{\text{max}}}$ from $\mathcal{S}_c = \{ i \in \mathcal{I} \mid c^*_i = \argmax_c a_{ic} \}$ that contains the cells $i$ that belongs to cluster $c$. The number of clusters are then found by computing $C = \left| \{ c \in \{1, \dots, C_{\text{max}}\} : \mathcal{S}_c \neq \emptyset \} \right|$.

That concludes the BCD algorithm. The complete algorithm from both blocks is provided in Algorithm~\ref{alg:BCD}, with update period $\tau$ to keep the model parameters and cluster assignments updated to the dynamic environment. In the next subsection, we study under which assumptions Algorithm~\ref{alg:BCD} converges.


\subsection{Convergence Analysis}
\label{sec:analysis}
To prove the convergence of Algorithm~\ref{alg:BCD}, we first have to make some standard assumptions in the literature~\cite{xu2013block, yu2019parallel}.

\begin{asum}\label{as:smooth}(\textbf{$L$-smoothness}) 
The function \(f(\mathbf{A},\boldsymbol{\Theta})\) in~\eqref{eq:f} is $L$-smooth if for every bounded set \(\mathcal B\), there exists a constant \(L_{\mathcal B} > 0\) for all \((\mathbf A,\boldsymbol\Theta),(\mathbf A',\boldsymbol\Theta')\in\mathcal B\) such that:
\begin{equation} \label{eq:a1}
\|\nabla f(\mathbf{A},\boldsymbol{\Theta}) - \nabla f(\mathbf{A}',\boldsymbol{\Theta}')\| \leq L_{\mathcal B}\, \|(\mathbf{A},\boldsymbol{\Theta}) - (\mathbf{A}', \boldsymbol{\Theta}')\|.
\end{equation}
\end{asum}

\begin{asum}\label{as:Lower}(\textbf{Bounded Loss Function}) 
There exists a constant $\ell^{*}$ such that $\ell\big(\boldsymbol{\omega}(\mathbf{x}_i; \boldsymbol{\theta}_c), \mathbf{y}_i (t+h)\big) \geq \ell^{*}$ for all $\boldsymbol{\theta}_c$ and $\mathbf{x}_i$. 
\end{asum}

\begin{asum}\label{as:gconv}(\textbf{Convexity}) 
The function $g:\mathbb{R}^{N\times C} \to \mathbb{R}\cup\{+\infty\}$ defined in~\eqref{eq:g} is convex, lower semi-continuous, and proper in the sense of convex analysis.
\end{asum}

\begin{asum}\label{as:block-updates}(\textbf{Block update parameter bounds})
The parameters \(L_\mathbf{A}^{k}\) used in the block steps satisfy uniform bounds
\[
0<\underline L_{\mathbf A}\le L_{\mathbf A}^k\le\overline L_{\mathbf{A}}<\infty\quad\text{for all }k.
\]
The same goes for the parameters \(L_\mathbf{\Theta}^{k}\).
\end{asum}

\begin{asum}\label{as:boundedness}(\textbf{Lower bounded})
The objective function \(F(\mathbf A,\boldsymbol\Theta)\) 
is bounded below on its effective domain:
\[
\inf_{(\mathbf A,\boldsymbol\Theta)} F(\mathbf A,\boldsymbol\Theta) > -\infty.
\]
Moreover, the initial point \((\mathbf A^0,\boldsymbol\Theta^0)\) satisfies \(F(\mathbf A^0,\boldsymbol\Theta^0)<+\infty\).
\end{asum}

From Assumption~\ref{as:gconv}, we need $g(\mathbf{A})$ in~\eqref{eq:g} to be convex, lower semi-continuous, and proper. Hence, we need to make sure that block 2 is convex.

\begin{proposition}\label{prop1}
Let Assumptions~\ref{as:smooth} and~\ref{as:boundedness} hold. Then, Assumption~\ref{as:gconv} holds if 
the matrix \(\mathbf{D}\) with entries
\(
D_{ij} = D\big(\mathcal{N}_i, \mathcal{N}_j\big)
\)
is \ac{PSD}.
\end{proposition}
\begin{proof}
    The proof of Proposition~\ref{prop1} is provided in Appendix.
\end{proof}

With Assumptions~\ref{as:smooth}--\ref{as:boundedness} and Proposition~\ref{prop1}, we are ready to present the main analytical result. In the following proposition,  we prove that Algorithm~\ref{alg:BCD} converges to a critical point of the function $F(\mathbf A,\boldsymbol\Theta)$.

\begin{proposition}
\label{prop:bcd-convergence}
Let $(\mathbf{A}^{k}, \boldsymbol{\Theta}^{k})_{k=1}^{\infty}$ be the sequence generated by Algorithm~\ref{alg:BCD}. Under Assumptions~\ref{as:smooth}-\ref{as:boundedness}, the sequence asymptotically converges to a critical point of the 
function $F(\mathbf{A}, \mathbf{\Theta})$, i.e., for $\tilde{\mathbf{A}} = \lim_{k\rightarrow \infty} \mathbf{A}^{k}$ and $\tilde{\mathbf{\Theta}} = \lim_{k\rightarrow \infty} \mathbf{\Theta}^{k}$, we have

    \begin{align}
        \partial L( \tilde{\mathbf{A}}, \tilde{\mathbf{\Theta}}) = 0.
    \end{align}
\end{proposition}

\begin{proof}
    The proof of Proposition~\ref{prop:bcd-convergence} is provided in Section~\ref{sec:appendix}.
\end{proof}

We note that Algorithm~\ref{alg:BCD} returns a sub-optimal solution to~\eqref{eq:relaxed_opt}. Since the optimization problem is an approximation of the original problem~\eqref{eq:opt}, the solution produced by Algorithm~\ref{alg:BCD} can be interpreted as an approximate, feasible solution to~\eqref{eq:opt}.


\subsection{Computational--Communication Complexity}
In the proposed clustered framework, we use FL~\cite{mcmahan2023communicationefficient} as follows: each \ac{BS} computes local model updates using its own dataset. These updates are periodically transmitted to a central server, which aggregates them to update the cluster-level models. The update period $\tau$ is a design parameter that trades off convergence speed against communication overhead. Once updated, the cluster-level models are redistributed to the \acp{BS} belonging to each cluster.

To contextualize the proposed approach, we compare it with a global approach in which a single model is trained on a central server, and a local approach in which each cell maintains and trains its own predictive model. In the local approach, all computation is performed on-site and no model parameters or gradients are exchanged, resulting in zero communication overhead. In contrast, both the global and our proposed clustered approaches require backhaul communication from the \acp{BS} to the server. Assuming the same update frequency of 15 minutes and number of local training epochs, the bandwidth consumption in the backhaul communication for the global model is $B_{\text{global}}^{\text{up}}=NFb=12\cdot 108 172 \cdot 4 \approx 5.2$MB, where $F$ is the number of model parameters and $b$ is the number of bytes that represent each scalar parameter. For the cluster approach, the same communication cost is $B_{\text{cluster}}^{\text{up}}=NF_{\ell}b=12\cdot 1548 \cdot 4 \approx 74$KB, where $F_\ell$ is the number of model parameters for the last layer. The numbers of model parameters from the MLP are calculated by simple arithmetic from the basics of neural networks. The ratio $F / F_\ell\approx70$ reveals that the upstream traffic is reduced by a factor of 70.

In the downstream communication of the global approach, a single updated model is broadcast to all \acp{BS} with a communication cost of $B_{\text{global}}^{\text{down}}=Fb=\cdot 108 172 \cdot 4 \approx 430$KB. 
In the clustered approach, $C$ distinct cluster models are multicast from the central server to the \acp{BS}. Consequently, the downlink communication cost scales linearly with the number of clusters $C$, leading to the communication cost $B_{\text{clustered}}^{\text{down}}=CF_\ell b=6\cdot 1548 \cdot 4 \approx 37$KB if we assume a constellation of 6 clusters.
This motivates keeping $C$ as small as possible while still capturing spatial heterogeneity in the network. We see that the ratio of the communication cost between the clustered approach and the global is  $F / (CF_\ell)\approx12$.

When it comes to the computational complexity, the global approach updates all model parameters, resulting in a per-round training complexity $\mathcal{O}(EFS_\text{train})$ where the total number of parameters $F$, the number of epochs $E$, and the total number of local training samples $S_\text{train}$. In the clustered approach, the computational complexity is $\mathcal{O}(EF_\ell S_\text{train} + NC^2)$ assuming that the number of epochs and total number of training samples are identical to the global approach. The $NC^2$ term comes from the \ac{SVD} in the clustered approach. Assuming the same update frequency, this leads to a reduction in arithmetic complexity in our proposed clustered approach, the global training. 
This reduction is particularly significant for deep models where the last layer constitutes a small fraction of the total number of parameters.

\section{Numerical Evaluation}
\label{sec:eval}
To evaluate the proposed framework, we generate a dataset using a combination of the ns-3 and Sionna simulators to obtain realistic data. The two upcoming subsections describe how the dataset is generated and present the numerical results of our predictive framework based on the simulated dataset.

\subsection{Experimental Setup}
\label{sec:setup}
To evaluate the proposed \ac{pQoS} framework, we conduct extensive wireless network simulations featuring high-mobility vehicular users in a realistic deployment scenario. For the purpose of reproducibility, the simulation environment is available in our public repository\footnote{https://github.com/osst3224/ns3-rt-mobility}.
The simulations capture both spatial and temporal variations in network conditions and build on a previous integration of the ns-3 network simulator~\cite{Riley2010} with the Sionna~\cite{hoydis2023sionnaopensourcelibrarynextgeneration} link-level ray-tracing engine~\cite{Pegu2505}. The Sionna environment enables the incorporation of realistic propagation characteristics into the ns-3 discrete network simulator by replacing the standard channel models with ray-tracing. 

We define vehicular mobility traces representative of typical city traffic patterns in the predefined Munich scene from Sionna. We deploy 30 connected vehicles along the defined mobility traces. Each vehicle acts as a \ac{UE}, connecting to the cell with the strongest signal. To reflect realistic variations in network load over a day, the traffic is modeled as a time-dependent function over 24 hours. For example, a low network traffic background intensity is applied during early morning hours, a moderate background load during daytime, and peak demand during rush hours with high background. We utilize a channel quality-aware scheduler that supplies throughput fairness. To further capture the effect of network heterogeneity, user mobility speeds are varied between 30 km/h and 100 km/h. This dynamic mobility profile allows us to analyze the robustness of the pQoS models under varying conditions. The 12 cells are placed at inter-site distances to cover the entire map with reasonable connectivity. The carrier frequency is set to 3.5 GHz with a bandwidth of 20 MHz, and the maximum transmit power of each cell is fixed at 30 dBm. 


To collect a dataset with network data from different layers, we integrate a FlowMonitor logging system with a sampling period of 1 second. The FlowMonitor is a network monitoring package within ns-3 that can retrieve samples from the simulation stack at regular intervals across all active flows. We record a set of network, mobility, and radio-level measurements that characterize the communication downlink conditions and the user context. The logs include the current timestamp, identifiers for the active data flow and serving cell, and the instantaneous load of that cell. User mobility is captured through the UE’s geographical position, speed, and movement direction. Packet-level statistics comprise the packet size, the number of packets transmitted from the cell, and the number of packets received at the UE during the sampling interval. The packet-level statistics also include average end-to-end latency, throughput, packet loss rate, and jitter over the sampling period. Radio conditions are described by the measured \ac{SINR}, \ac{RSRP}, and an indicator of line-of-sight or non-line-of-sight propagation. 

The predictive model is implemented as a \ac{MLP} designed to predict the future multidimensional distribution of key QoS indicators within a given cell over the prediction horizon of 1 hour. For each prediction cycle, the neural network receives as input all network\mbox{-}, mobility\mbox{-}, packet\mbox{-}, and radio-level measurements collected during the most recent sampling period. These features provide a detailed description of the instantaneous operating conditions, reflecting traffic load, user movement, packet performance, and radio propagation quality. To incorporate a long-term temporal context, the input also includes the parameters of a fitted multidimensional Gaussian distribution that summarizes the joint behavior of latency, jitter, and RSRP in the same cell over the preceding 15-minute period. Together, these inputs represent both short-term dynamics and historical trends that are relevant for a long-term prediction horizon.
The output of the neural network is the predicted joint distribution of latency, jitter, and RSRP 1 hour into the future for the same cell. 
The dataset for each cell and hour, containing around 10,000 samples, is split into 70\% training, 10\% validation, and 20\% test sets. The predictive \ac{MLP} is implemented with three hidden layers of 256, 256, and 128 neurons using ReLU activations and 20\% dropout. The model is trained using the Adam optimizer with learning rate $10~{-3}$, batch size 32, \ac{MSE} loss, and a fixed random seed for reproducibility. Furthermore, we set the update period $\tau=15$ minutes.

To measure the distance between two Gaussian distributions, we go for the Hellinger kernel, which is \ac{PSD}~\cite{hkernel}. The Hellinger kernel $k(\mathcal{N}_i, \mathcal{N}_j)$ has the closed loop form 
\begin{equation*}
\text{exp}\left(\frac{(\det\boldsymbol{\Sigma}_i \det\boldsymbol{\Sigma}_j)^{1/4}}{\det\left(\frac{\boldsymbol{\Sigma}_i + \boldsymbol{\Sigma}_j}{2}\right)^{1/2}} e^{-\frac{1}{8} (\boldsymbol{\mu}_i - \boldsymbol{\mu}_j)^T \left(\frac{\boldsymbol{\Sigma}_i + \boldsymbol{\Sigma}_j}{2}\right)^{-1} (\boldsymbol{\mu}_i - \boldsymbol{\mu}_j)}\right),
\end{equation*}
in the special case where $\mathcal{N}_i$ and $\mathcal{N}_j$ are Gaussians. To keep the \ac{PSD} properties of the Hellinger kernel, and keep the properties of our optimization problem so that we still want to minimize~\eqref{eq:relaxed_opt} with respect to $\mathbf{A}$, we from the Laplacian $\mathbf{D}=\text{diag}(\mathbf{K}\mathbbm{1})-\mathbf{K}$ to represent our distance metric $\mathbf{D}_{ij} =  D(\mathcal{N}_i, \mathcal{N}_j)$. The fact that $\mathbf{K}$ is symmetric and non-negative ensures that $\mathbf{D}$ is PSD~\cite{von2007tutorial}. This is important to be aligned with Assumption~\ref{as:gconv}. 



\begin{figure}
    \centering
    \includegraphics[width=0.69\linewidth]{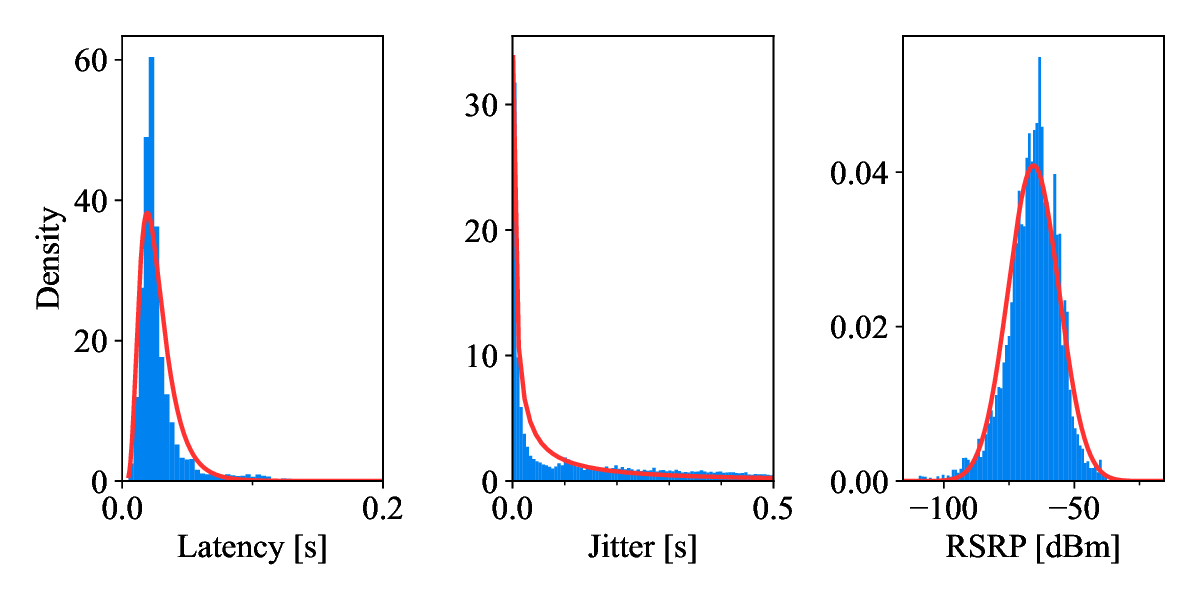}
    \caption{Empirical histograms of latency, jitter, and RSRP together with their fitted probability density functions. Latency and jitter are well approximated by lognormal distributions, while RSRP follows a Gaussian distribution in the logarithmic (dBm) domain, indicating a good agreement between the empirical data and the assumption of lognormal distributions.}
    \label{fig:lognormal}
\end{figure}

Finally, we need to address a practical part of the solution approach. The covariance of a Gaussian distribution is \ac{PSD}. However, even though the output in the training data only contains \ac{PSD} covariance matrices, the output of the MLP does not ensure that the output forms a \ac{PSD} covariance matrix. To ensure the predicted covariance matrix $\boldsymbol{\Sigma}_{(i, h)}$ remains \ac{PSD}, we utilize a Cholesky parameterization~\cite{Cholesky_article} to generate the vector $\mathbf{l}_i\in\mathbb{R}^{d(d+1)/2}$ representing the non-zero elements of a lower-triangular matrix $\mathbf{L}_i$. 
For each element $l_{ij}$ that lies in the diagonal of the matrix $\mathbf{L}_i$, we apply the mapping $l_{ij} = \exp(\tilde{l}_{ij})$, where $\mathbf{\tilde{l}}_{i}$ is the unconstrained vector generated by the predictive model. This ensures that the diagonal values in the matrix $\mathbf{L}_i$ meets the constraint $l_{ij} > 0$, which is a necessary and sufficient condition for the resulting covariance matrix $\boldsymbol{\Sigma}_{(i,h)} = \mathbf{L}_i \mathbf{L}_i^\top$ to be non-singular and strictly positive definite.

Hence, we train a predictive model for each cluster that predicts the parameters \( \mathbf{y}_i (t+h) =  ([\boldsymbol{\mu}_i, \mathbf{\tilde{l}}_{i}]) \in {R}^{d+d(d + 1)/2} \). As described in the previous paragraph, the parameters in vector $\mathbf{\tilde{l}}_{i}$ are used to obtain the covariance matrix $\boldsymbol{\Sigma}_{(i,h)} = \mathbf{L}_i \mathbf{L}_i^\top$ for the future multivariate distribution \(\mathbf{Q}_i (t+h) \sim \mathcal{N}( \boldsymbol{\mu}_{(i,\, h)}, \boldsymbol{\Sigma}_{(i,\, h)}) \).

\subsection{Results}




We assess the proposed clustered \ac{FL} framework and the BCD solver on a dataset generated from the experimental setup in the previous subsection to analyze convergence behavior, cluster evolution during optimization, and prediction performance across cells. To verify the distributional assumptions of the considered QoS metrics, we plot the empirical data in Fig.~\ref{fig:lognormal}. The fitted lognormal distributions show good agreement with the observed data, supporting the use of lognormal models for the QoS metrics.

Initially, we train a global neural network model using the combined training data from all cells, thereby capturing the best achievable performance of a model learned from the entire region. This global predictive neural network is implemented as the predictive model in Algorithm~\ref{alg:BCD}, which we apply transfer learning to, meaning that we only train the last layer of the global neural network. 

The optimization objective in~\eqref{eq:relaxed_opt} produced by the alternating BCD updates decreases steadily and stabilizes with the number of iterations. The convergence trend is presented in Fig.~\ref{fig:opt_loss}, indicating that the alternating cluster reassignments and model updates rapidly capture the available gains. We present two solutions obtained under different smoothness conditions: one that satisfies Assumption~\ref{as:smooth} on $L$-smoothness and one that violates it. The violation is induced by replacing the linear activation function with the non-smooth ReLU activation. As illustrated in Fig.~\ref{fig:opt_loss}, the solution satisfying the smoothness assumption exhibits a smoother behavior, whereas the non-smooth formulation leads to a more irregular and unstable solution.

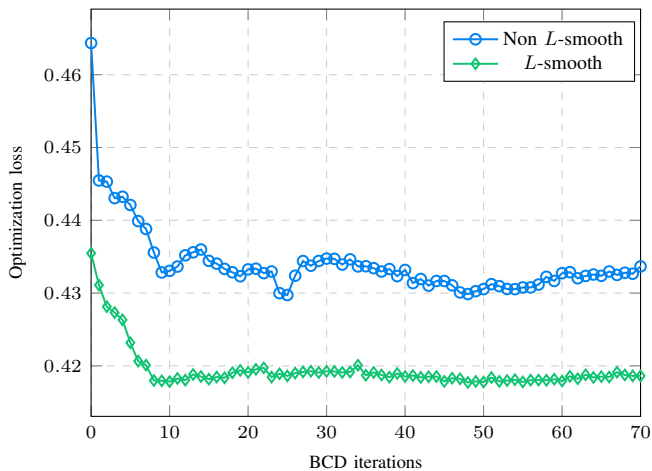
\begin{figure}[t]
\centering
    \begin{tikzpicture} 
    \begin{axis}[
        xlabel={BCD iterations},
        ylabel={Optimization loss},
        label style={font=\scriptsize},
        tick label style={font=\scriptsize} , 
        width=0.49\textwidth,
        height=7cm,
        xmin=0, xmax=70,
        legend style={nodes={scale=0.75, transform shape}, at={(0.28,0.425)}, at={(0.99,0.97)}}, 
        ymajorgrids=true,
        xmajorgrids=true,
        grid style=dashed,
        grid=both,
        grid style={line width=.1pt, draw=gray!15},
        major grid style={line width=.2pt,draw=gray!40},
    ]

    \addplot[
        color=eblue,
        mark=o,
        mark options = {rotate = 180},
        line width=0.75pt,
        mark size=2pt,
        ]
    table[x=Column1,y={opt_loss}]
    {Figures/Data/loss-lambda0_03-eta0_1-beta0_005_r.csv};

    \addplot[
        color=egreen,
        mark=diamond,
        mark options = {rotate = 180},
        line width=0.75pt,
        mark size=2pt,
        ]
    table[x=Column1,y={opt_loss}]
    {Figures/Data/loss-lambda0_03-eta0_1-beta0_005_l.csv};

    \addlegendentry{Non $L$-smooth}
    \addlegendentry{$L$-smooth}
    
    \end{axis}
\end{tikzpicture}
  \caption{The optimization loss $F(\mathbf A,\boldsymbol\Theta)$ of the relaxed optimization problem in~\eqref{eq:relaxed_opt} as a function of the number iterations $k$ in  Algorithm~\ref{alg:BCD}. Two solutions are included, one with Assumption~\ref{as:smooth} of $L$-smoothness intact, and one violating Assumption~\ref{as:smooth}. As shown in this example, the algorithm converges to a stationary point as the number of iterations increases.}
  \label{fig:opt_loss}
\end{figure}


The cluster reassignments initially affect the optimization loss rapidly, as can be seen up until 10 iterations in Fig.~\ref{fig:opt_loss}. Small, rapid movements after the 10th iteration also mark cluster reassignments. As the clusters converge to a stable configuration, so does the predictive error and thus also the optimization loss in~\eqref{eq:relaxed_opt}.


The clustering dynamics affect how the algorithm adapts model granularity to the data. Early iterations generally produce a larger number of clusters, while later iterations merge statistically similar groups and converge to a configuration with fewer clusters. This transition from six initial clusters down to four final clusters of cells is visualized in Fig.~\ref{fig:cluster_constellation} as a stacked bar chart. This reflects the method’s ability to eliminate redundant models while preserving coverage of heterogeneous operating cells.

We compare the predictive accuracy of the clustered scheme against two baselines, a single global model and fully local models trained per cell. The violin plot in Fig.~\ref{fig:violin_global} shows distributions of the predicted error of the expected value of the latency in milliseconds for all 12 cells. The plot demonstrates that the clustered approach yields lower prediction error and reduced variance of the prediction error relative to the global baseline, indicating improved robustness to spatial heterogeneity. The global model represents the pre-trained neural network used to initialize the clustered predictors. The improved performance of the clustered predictors thus motivates the use of transfer learning for capturing cluster-level heterogeneity.

When comparing the same prediction error of the clustered approach to the per-cell local baseline in Fig.~\ref{fig:violin_local}, the two predictive models produce similar error distributions. Although in some cases, in cells 3, 6, 9, and 11, the local model obtains a far higher variance in the predicted error. This suggests that cluster-level predictors capture the local variability while using far fewer distinct models, especially in large-scale deployments.

\begin{figure}[t]
    \centering
    \includegraphics[width=0.53\linewidth, clip=True, trim={0cm 0cm 0cm 0.3cm}]{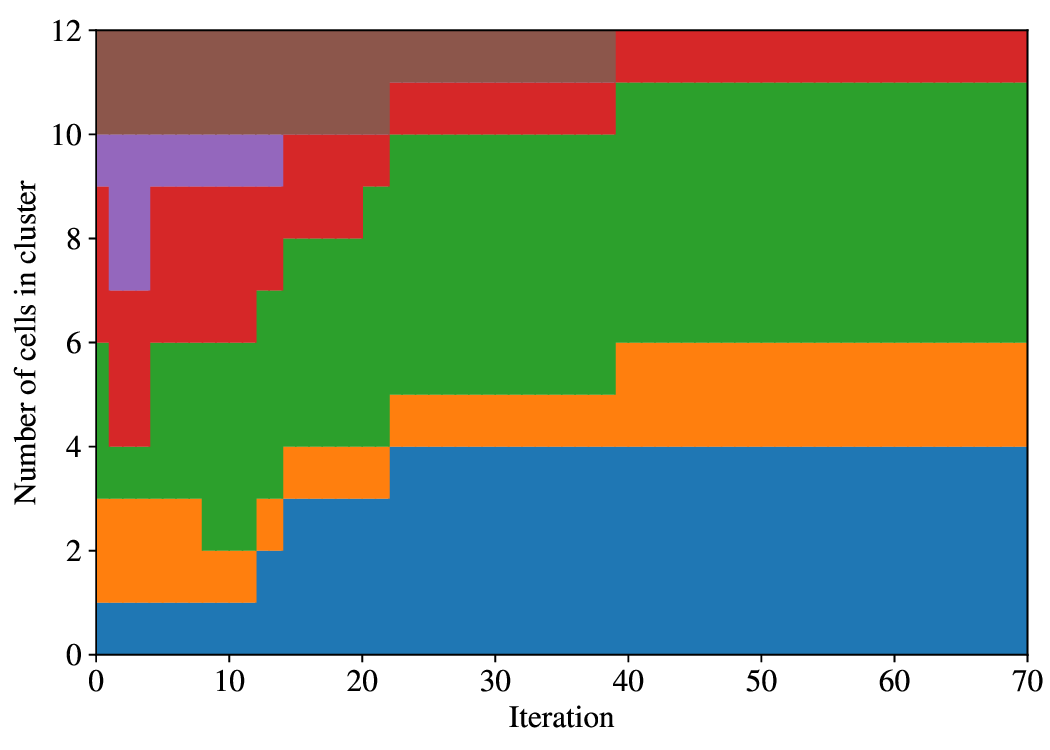}
    \caption{A visualization of how the cluster constellation changes with the number of iterations of the proposed solution framework. The colors represent different clusters. In this case, the algorithm converges towards clusters with 5, 4, 2, and 1 cells, respectively.}
    \label{fig:cluster_constellation}
\end{figure}

\begin{figure}[t]
    \centering

    \subfloat[Global vs. Clustered approaches.\label{fig:violin_global}]{
        \begin{minipage}{0.49\linewidth}
            \centering
            \includegraphics[width=\linewidth]{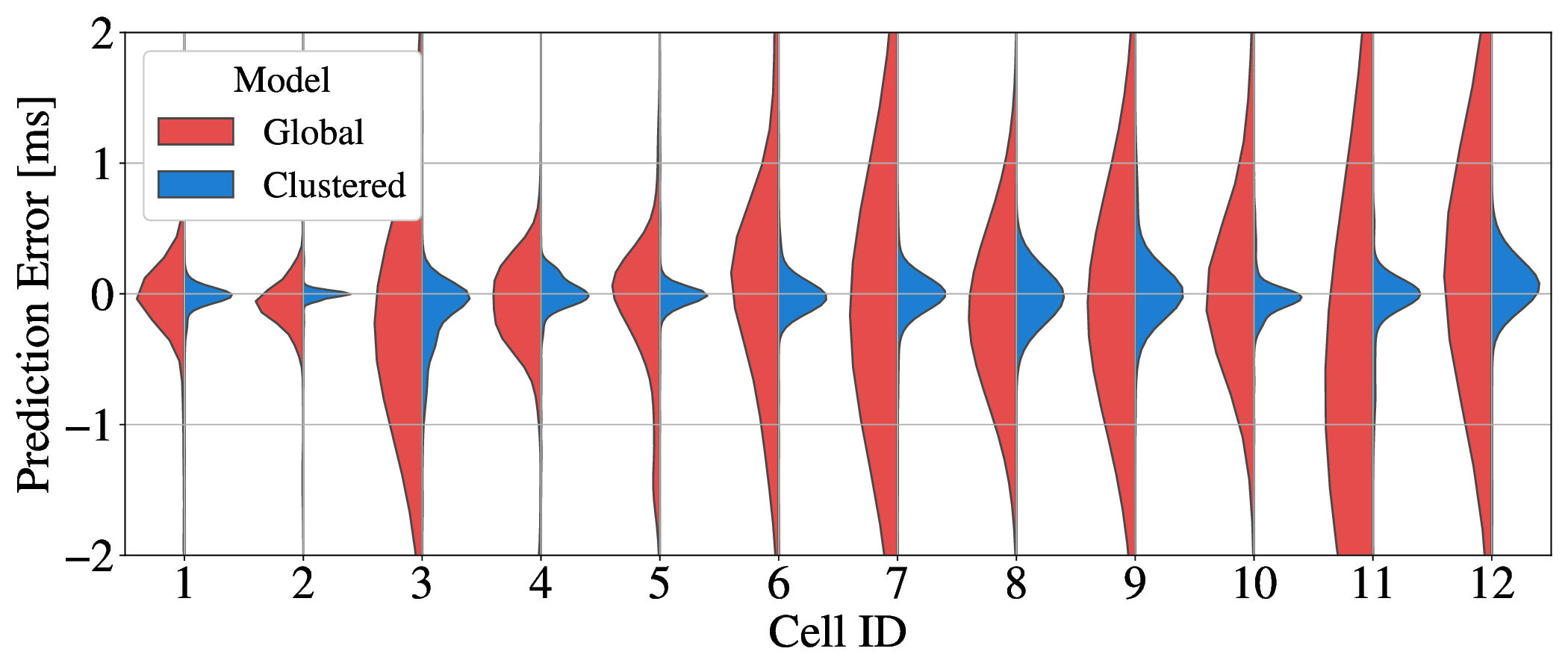}
        \end{minipage}
    }
    \subfloat[Local vs. Clustered approaches. \label{fig:violin_local}]{
        \begin{minipage}{0.49\linewidth}
            \centering
            \includegraphics[width=\linewidth]{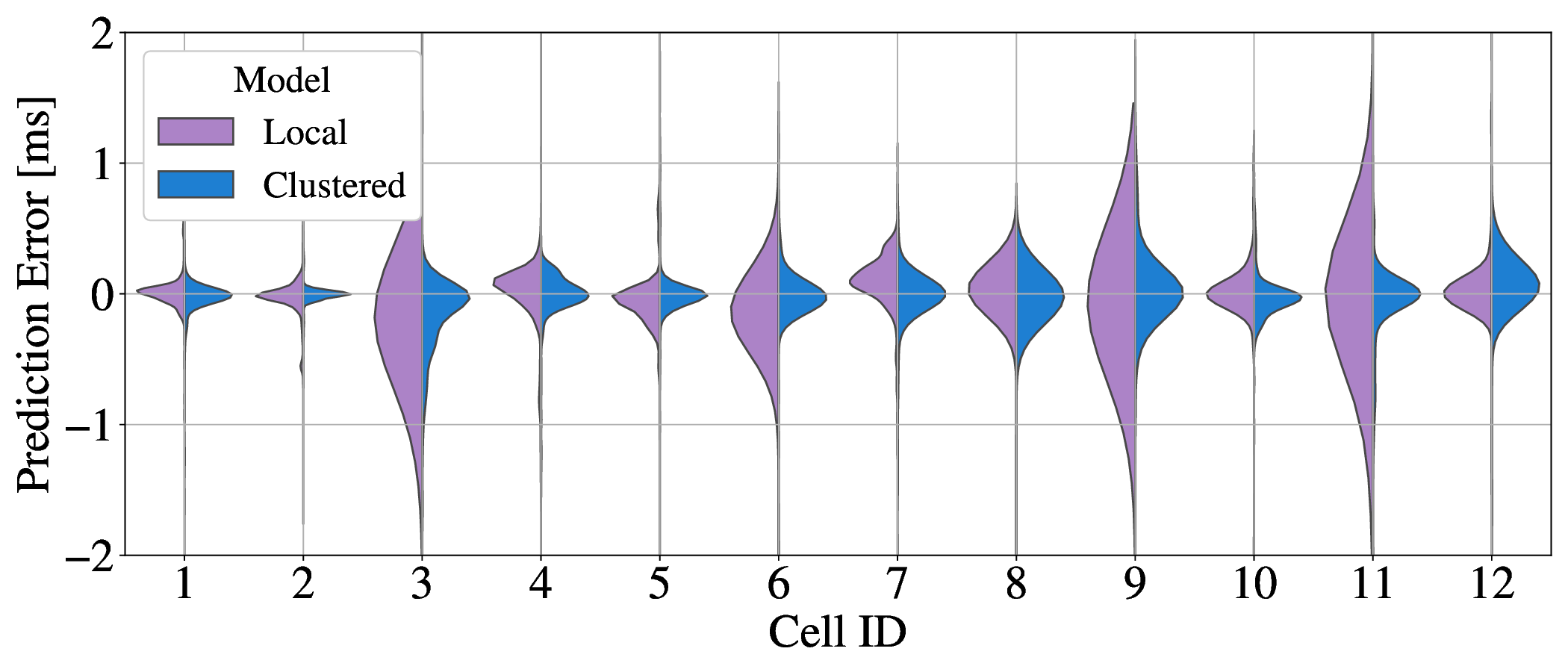}
        \end{minipage}
    }
    \caption{Distribution of prediction errors for the mean latency (ms) across different cells. The plots compare the proposed clustered approach against global and local baselines. Note that the heights of the distributions have been normalized for visibility.}
    \label{fig:violin_combined}
    
\end{figure}

\begin{table}[ht]
\centering
\caption{The mean absolute error (MAE) of predicted Gaussian distribution parameters.}
\label{tab:norm}
\begin{tabular}{llccc}
\hline
\textbf{Parameter} \rule[-1.2ex]{0pt}{0pt}  & \textbf{Method} & \makecell{\textbf{Latency} \\ \textbf{[ms]}} & \makecell{\textbf{Jitter} \\ \textbf{[ms]}} & \makecell{\textbf{RSRP} \\ \textbf{[dBm]}} \\
\hline
\multirow{3}{*}{Mean}
 & Global    & 0.59 & 3.9 & 0.49 \\
 & Local     & 0.16 & \textbf{1.1} & \textbf{0.13} \\
 & Clustered & \textbf{0.12} & 1.2 & 0.16 \\
\hline
\multirow{3}{*}{Std. Dev.}
 & Global    & 0.32 & 3.1 & 0.36 \\
 & Local     & 0.11 & 1.1 & 0.11 \\
 & Clustered & \textbf{0.08} & \textbf{0.97} & \textbf{0.10} \\
\hline
\end{tabular}
\end{table}

A wider evaluation of the predictive accuracy of our proposed framework, compared to the global and local baseline models, is provided in Table~\ref{tab:norm}. This table summarizes the \ac{MAE} of the predicted Gaussian distribution parameters of latency, jitter, and RSRP. Across these KPIs, the clustered models consistently reduce the MAE relative to the global baseline traditionally used in previous work. The local and clustered models achieve comparable accuracy in predicting the mean of the Gaussian distribution across the considered QoS metrics, with differences in MAE below 0.04 ms for latency and 0.03 dB for RSRP. However, the clustered model consistently improves the prediction of the standard deviation, achieving MAE reductions of approximately 27\% for latency (from 0.11 to 0.08), 12\% for jitter (from 1.1 to 0.97), and 9\% for RSRP (from 0.11 to 0.10). This suggests that clustering enables more robust prediction of second-order statistics by leveraging shared structure across cells with similar operating conditions. This aligns with the results in Fig.~\ref{fig:violin_local}. The clustered approach adapts to spatiotemporal traffic load variations across the network induced in the dataset, described in Section~\ref{sec:setup}. As a result, the clustered approach provides more reliable uncertainty characterization than purely local training while maintaining similar mean prediction performance.

\begin{table}[ht]
\centering
\caption{Negative log-likelihood (NLL) comparison across global, local, and clustered methods.}
\label{tab:nll_results}
\begin{tabular}{lccc|c}
\hline
\textbf{Method} & \makecell{\textbf{Latency}} & \makecell{\textbf{Jitter}} & \makecell{\textbf{RSRP}} & \makecell{\textbf{Total NLL}} \\
\hline
Global    & -3.4184 & 0.1755 & 4.6419 & 1.3990 \\
Local     & -3.4265 & 0.1753 & \textbf{4.6415} & 1.3902 \\
Clustered & \textbf{-3.4288} & \textbf{0.1745} & 4.6424 & \textbf{1.3881} \\
\hline
\end{tabular}
\end{table}

Table~\ref{tab:nll_results} presents the \ac{NLL} of the prediction results, demonstrating that the clustered approach achieves the best probabilistic calibration, yielding the lowest total NLL. This further adds to the argument that by grouping behaviorally similar cells, the model effectively improves the global model and resolves the data-heterogeneity issues of the local approach.

While the numerical gains may appear marginal, the NLL is a logarithmic scale metric. Thus, small absolute reductions reflect exponential improvements in predictive likelihood. In large-scale deployments involving thousands of \acp{BS}, these consistent gains in modeling latency, jitter, and other QoS metrics translate to a more reliable network representation.

\subsection{Optimality Gap Analysis}
To assess the tightness of the proposed relaxation, we evaluate the optimality gap between the relaxed and original combinatorial formulations for the problem with $N = 12$. The relaxed problem is solved using the proposed \ac{BCD} algorithm, while the original problem is solved exactly via exhaustive search using dynamic programming. 

The optimality gap is defined as
\begin{equation}
\text{gap} = \frac{f_{\text{orig}}^* - f_{\text{rel}}^*}{f_{\text{orig}}^*} \times 100\%,
\end{equation}
where $f_{\text{rel}}^*$ and $f_{\text{orig}}^*$ denote the optimal objective values of the relaxed and original problems, respectively.

For the configuration with $N=12$, 
the observed gap reaches approximately $23\%$, presented by the $L$-smooth solution in Fig.~\ref{fig:opt_loss}. This indicates that, while the relaxation provides a valid lower bound, it is not tight for these problem instances, and the BCD solver may converge to solutions that deviate from the true combinatorial optimum. Nevertheless, the relaxed formulation remains computationally attractive and is extremely useful when exact optimization is infeasible, particularly for large-scale settings where exhaustive search is impractical. While the gap reached roughly 23\% on small instances, this is expected for NP-hard 
problems and does not preclude the method from offering effective and scalable approximate solutions for large datasets.

Although the solution with an $L$-smooth loss in Fig.~\ref{fig:opt_loss} provides more stable convergence behavior, the solution in Fig.~\ref{fig:opt_loss} with a non-$L$-smooth loss attains a solution closer to the combinatorial optimum with an optimality gap of $21\%$. This observation does not contradict the theoretical guarantees, which pertain to convergence stability rather than approximation tightness as presented in Fig.~\ref{fig:opt_loss}.


\section{Discussion}
\label{sec:discusion}

The results demonstrate that the joint clustering and training BCD-based framework produces stable cluster assignments and distributional forecasts that substantially improve the accuracy of a global model, while improving the predictive accuracy of the second-order statistics of per-cell models. The clustered approach is resource-efficient as it requires fewer predictive models and enables cells to migrate to a model that better reflects their current wireless conditions, instead of retraining new models. The framework achieves a favorable trade-off between predictive accuracy and model complexity. Ultimately, this makes it well-suited for SLA-driven networks in dynamic environments, where accurate long-horizon QoS forecasts are essential for making informed network management decisions and meeting reliability commitments. 

The results can be considered statistically representative because the dataset used for evaluation includes the main sources of variability in a realistic vehicular network scenario. The combined ns-3 and Sionna simulation setup introduces ray-traced propagation, time-varying traffic load over 24 hours, and heterogeneous mobility patterns, ensuring diverse operating conditions. Moreover, the 1-second sampling interval provides a large number of observations per cell, leading to stable error distributions. The consistency between the multiple metrics evaluated further indicates that the observed improvements are systematic rather than due to sampling noise. The alignment of these independent indicators provides strong evidence that the results are statistically robust.

Future work should consider scaling the framework to larger cellular deployments and incorporating additional KPI dimensions, such as uplink performance. 
Such a study would have the possibility to investigate the positive effects of using our framework to let cells migrate to different cells between training rounds. 
Finally, integrating the framework into an online network with SLA management is an important step toward deployment in commercial networks.

\section{Conclusion}
\label{sec:conclusion}
We presented a joint clustering and training framework based on \ac{BCD} that learns cluster-specific predictive QoS models while jointly optimizing cell assignments to the clusters. Evaluated on a realistic dataset, generated by a simulator combining ns-3 and Sionna, the method consistently improves prediction accuracy over a single global model and achieves performance comparable to per-cell models while using far fewer stored predictors. The approach offers a practical trade-off between predictive accuracy and model complexity, which is attractive for SLA-driven deployments where resource-efficient forecasts are required. Future work should investigate online adaptation of the proposed framework in networks with SLA management systems.

\appendix
\label{sec:appendix}

\subsection{Proof of Proposition~\ref{prop1}}
\label{sec:proof_prop1} 
The function \( g(\mathbf{A}) \) does not take the value \( +\infty \) anywhere in its domain. Furthermore, since the function $g(\mathbf{A})$ can be written in quadratic form, it is finite for all finite \( \mathbf{A} \). If the distance matrix \( \mathbf{D} \) is \ac{PSD}, it is ensured that the quadratic form is well-defined and nonnegative. For these reasons, the function \( g(\mathbf{A}) \) is proper. 

If $g(\mathbf{A})$ is convex, it is lower semicontinuous. For that reason, we prove that $g(\mathbf{A})$ is convex.

\begin{proof}
We rewrite \( g(\mathbf{A}) \) on quadratic form:

\begin{equation}
    g(\mathbf{A}) = 
    \lambda \sum_{k=1}^{C} \boldsymbol{a}_c^T \mathbf{D} \boldsymbol{a}_c.
\end{equation}
We proceed by computing the gradient with respect to \( a_{ic} \) as

\begin{equation}
    \frac{\partial g}{\partial a_{ic}} = 2\lambda \sum_{j=1}^{N} a_{jc} D_{ij}.
\end{equation}
Then, we compute the Hessian by with respect to \( a_{i'c'} \):
\begin{equation}
    \frac{\partial^2 g}{\partial a_{ic} \partial a_{i'c'}} = 2\lambda D_{ii'} \delta_{cc'}.
\end{equation}
Here, \( D_{ii'} \) represents the distance between \( \mathcal{N}_i \) and \( \mathcal{N}_{i'} \), and \( \delta_{cc'} \) is the Kronecker delta, which equals 1 if \( c = c' \) and 0 otherwise.

The Hessian matrix \(\mathbf{H} = \nabla^2 g(\mathbf{A}) \) consists of blocks corresponding to different clusters \( c \). Each block \( \mathbf{H}_c \) corresponds to a separate submatrix in the Hessian. The off-diagonal blocks 
are zero because of the Kronecker delta \( \delta_{cc'} \), meaning there is no interaction between different clusters.
Thus, the Hessian is block diagonal:
\[
\nabla^2 g(\mathbf{A}) = \mathrm{diag}(\mathbf{H}_1, \ldots, \mathbf{H}_C), \quad \mathbf{H}_c = 2\lambda \mathbf{D}.
\]
Since we have already assumed that the distance matrix \( \mathbf{D} \) is positive semidefinite, each Hessian block \( \mathbf{H}_c \) is \ac{PSD}. A block diagonal matrix is \ac{PSD} if all its diagonal blocks are \ac{PSD}. This implies that the entire Hessian is positive semidefinite, which guarantees that \( g(\mathbf{A}) \) is convex.
\end{proof}

\subsection{Proof That Assumption~\ref{as:smooth} Holds}
\begin{lem}
The function $f(\mathbf{A}, \boldsymbol{\Theta})$ in~\eqref{eq:f}, with smooth activation functions in the neural network function \(\boldsymbol{\omega}(\mathbf{x}_i; \boldsymbol{\theta}_c)\) and smooth loss function \( \ell \), is L-smooth with a constant \( L > 0 \), satisfying:
\begin{equation}
    \|\nabla f(\mathbf{A}, \boldsymbol{\Theta}) - \nabla f(\mathbf{A}', \boldsymbol{\Theta}')\| \leq L \|(\mathbf{A}, \boldsymbol{\Theta}) - (\mathbf{A}', \boldsymbol{\Theta}')\|.
\end{equation}
\end{lem}

\begin{proof}
The function \( f(\mathbf{A}, \boldsymbol{\Theta}) \) is a weighted sum of loss functions \( \ell(\boldsymbol{\omega}(\mathbf{x}_i; \boldsymbol{\theta}_c), \mathbf{y}_i(t+h)) \), multiplied by weights \( a_{ic} \), which are fixed nonnegative weights that sum to 1.
The differentiability of \( f(\mathbf{A}, \boldsymbol{\Theta}) \) thus depends on the smoothness properties of the weights $\boldsymbol{a}$, the loss function \( \ell \) and the neural network function \( \boldsymbol{\omega}(\mathbf{x}_i; \boldsymbol{\theta}_c)\).

First, we evaluate the gradient $\nabla_{\mathbf{A}}$ of the quadratic form of \(f(\mathbf{A}, \boldsymbol{\Theta})\) from~\eqref{eq:quadratic}.
We calculate the gradient for each $\nabla_{\boldsymbol{a}_{c}}$ and find that:
\begin{equation}
    \nabla_{\boldsymbol{a}_{c}} f(\mathbf{A}, \boldsymbol{\Theta}) = \sum_{c=1}^{C} \left(\boldsymbol{\ell}_c + 2\lambda\, \mathbf{D}\, \boldsymbol{a}_{c}\right).
    \label{eq:grad_diff_w}
\end{equation}
Since \( \ell \) is smooth, there exists \( L_\ell \) such that:
\begin{equation}\label{eq:grad_diff_ell}
    |\ell(\boldsymbol{\omega}(\mathbf{x}_i; \boldsymbol{\theta}_c), \mathbf{y}_i) - \ell(\boldsymbol{\omega}(\mathbf{x}_i; \boldsymbol{\theta}_c'), \mathbf{y}_i)| \leq L_\ell \|\boldsymbol{\theta}_c - \boldsymbol{\theta}_c'\|.
\end{equation}
To continue, we check the second term. By factorizing the second term as
\begin{equation}
\| 2\lambda \mathbf{D} (\boldsymbol{a}_c - \boldsymbol{a}_c') \| \leq 2\lambda \|\mathbf{D}\| \|\boldsymbol{a}_c - \boldsymbol{a}_c'\|=L_\mathbf{D}\|\boldsymbol{a}_c - \boldsymbol{a}_c'\| .
\label{eq:grad_diff_D}
\end{equation}
Thus, we find that the function $f(\mathbf{A},\boldsymbol{\Theta})$ is L-smooth with respect to the gradient $\nabla_{\mathbf{A}}$.

Next, we focus on the difference for the gradient $\nabla_{\boldsymbol{\Theta}}$ of  \(f(\mathbf{A}, \boldsymbol{\Theta})\).
Since \( f(\mathbf{A}, \boldsymbol{\Theta}) \) is separable across clusters \( c \), we compute:

\begin{equation}
    \nabla_{\boldsymbol{\theta}_c} f(\mathbf{A}, \boldsymbol{\Theta}) = \sum_{i=1}^{N} a_{ic} \nabla_{\boldsymbol{\theta}_c} \ell\Big(\boldsymbol{\omega}(\mathbf{x}_i; \boldsymbol{\theta}_c), \mathbf{y}_i(t+h)\Big).
\end{equation}

By applying the chain rule:

\begin{equation} \label{eq:chain_rule}
    \nabla_{\boldsymbol{\theta}_c} \ell(\boldsymbol{\omega}(\mathbf{x}_i; \boldsymbol{\theta}_c), \mathbf{y}_i) = \frac{\partial \ell}{\partial a} \Big|_{a = \boldsymbol{\omega}(\mathbf{x}_i; \boldsymbol{\theta}_c)} \cdot \nabla_{\boldsymbol{\theta}_c} \boldsymbol{\omega}(\mathbf{x}_i; \boldsymbol{\theta}_c).
\end{equation}

We define:
\begin{itemize}
    \item \( h_i(\boldsymbol{\theta}_c) = \frac{\partial \ell}{\partial a} \Big|_{a = \boldsymbol{\omega}(\mathbf{x}_i; \boldsymbol{\theta}_c)} \),
    \item \( J_i(\boldsymbol{\theta}_c) = \nabla_{\boldsymbol{\theta}_c} \boldsymbol{\omega}(\mathbf{x}_i; \boldsymbol{\theta}_c) \). 
\end{itemize}
This gives us:
\begin{equation}
    \nabla_{\boldsymbol{\theta}_c} f(\mathbf{A}, \boldsymbol{\Theta}) = \sum_{i=1}^{N} a_{ic} h_i(\boldsymbol{\theta}_c) J_i(\boldsymbol{\theta}_c).
    \label{f_grad_theta}
\end{equation}

First off, to prove smoothness, we show that the difference for 
\begin{equation}
    \|\nabla_{\boldsymbol{\theta}_c} f(\mathbf{A}, \boldsymbol{\Theta}) - \nabla_{\boldsymbol{\theta}_c} f(\mathbf{A}', \boldsymbol{\Theta}')\|,
    \label{eq:f_grad_theta_diff}
\end{equation}
is bounded. By inserting the gradient expression in~\eqref{f_grad_theta} into~\eqref{eq:f_grad_theta_diff}, we get
\begin{equation}
    \left\| \sum_{i=1}^{N} a_{ic} h_i(\boldsymbol{\theta}_c) J_i(\boldsymbol{\theta}_c) - a_{ic}' h_i(\boldsymbol{\theta}_c') J_i(\boldsymbol{\theta}_c') \right\|.
\end{equation}
By adding and subtracting \( a_{ic}' h_i(\boldsymbol{\theta}_c) J_i(\boldsymbol{\theta}_c) \), we get:
\begin{equation} \label{eq:grad_diff}
\begin{aligned}
    & \|\nabla_{\boldsymbol{\theta}_c} f(\mathbf{A}, \boldsymbol{\Theta}) - \nabla_{\boldsymbol{\theta}_c} f(\mathbf{A}', \boldsymbol{\Theta}')\|  \\
    &\quad \leq \sum_{i=1}^{N} \Big[ |a_{ic} - a_{ic}'| \cdot \| h_i(\boldsymbol{\theta}_c) J_i(\boldsymbol{\theta}_c) \|  \\
    &\quad \quad + a_{ic}' \| h_i(\boldsymbol{\theta}_c) J_i(\boldsymbol{\theta}_c) - h_i(\boldsymbol{\theta}_c') J_i(\boldsymbol{\theta}_c') \| \Big].
\end{aligned}
\end{equation}

We will show that the two terms in the sum are bounded separately. 
Because the neural networks have bounded activations and derivatives on compact domains, \( \| h_i(\boldsymbol{\theta}_c) J_i(\boldsymbol{\theta}_c) \| \) is bounded. This implies that there exists a constant \( M_J \) such that $\| h_i(\boldsymbol{\theta}_c) J_i(\boldsymbol{\theta}_c) \| \leq M_J$.
Thus, we get:
\begin{equation}
    |a_{ic} - a_{ic}'| \cdot \| h_i(\boldsymbol{\theta}_c) J_i(\boldsymbol{\theta}_c) \| \leq M_J |a_{ic} - a_{ic}'|.
\end{equation}

Next, we address the second term in the sum in~\eqref{eq:grad_diff}.
Since both \( h_i(\boldsymbol{\theta}_c) \) and \( J_i(\boldsymbol{\theta}_c) \) are smooth, there exist constants \( L_h \) and \( L_J \) such that:
\begin{equation}
    \| h_i(\boldsymbol{\theta}_c) - h_i(\boldsymbol{\theta}_c') \| \leq L_h \|\boldsymbol{\theta}_c - \boldsymbol{\theta}_c'\|,
\end{equation}
\begin{equation}
    \| J_i(\boldsymbol{\theta}_c) - J_i(\boldsymbol{\theta}_c') \| \leq L_J \|\boldsymbol{\theta}_c - \boldsymbol{\theta}_c'\|.
\end{equation}
Using the triangle inequality:
\begin{equation}
\begin{aligned}
    \| h_i(\boldsymbol{\theta}_c) J_i(\boldsymbol{\theta}_c) & - h_i(\boldsymbol{\theta}_c') J_i(\boldsymbol{\theta}_c') \| \\ & \leq | h_i(\boldsymbol{\theta}_c) - h_i(\boldsymbol{\theta}_c') | \cdot \| J_i(\boldsymbol{\theta}_c) \|  \\
    & + | h_i(\boldsymbol{\theta}_c') | \cdot \| J_i(\boldsymbol{\theta}_c) - J_i(\boldsymbol{\theta}_c') \|.
\end{aligned}
\end{equation}

Since \( \| J_i(\boldsymbol{\theta}_c) \| \leq M_J \) and \( | h_i(\boldsymbol{\theta}_c') | \leq M_g \), we obtain:

\begin{equation}
    \| h_i(\boldsymbol{\theta}_c) J_i(\boldsymbol{\theta}_c) - h_i(\boldsymbol{\theta}_c') J_i(\boldsymbol{\theta}_c') \| \leq (L_h M_J + M_g L_J) \|\boldsymbol{\theta}_c - \boldsymbol{\theta}_c'\|.
\end{equation}

Substituting the bounds into~\eqref{eq:grad_diff}, we get:
\begin{equation} \label{eq:grad_diff1}
\begin{aligned}
    & \|\nabla_{\boldsymbol{\theta}_c} f(\mathbf{A}, \boldsymbol{\Theta}) - \nabla_{\boldsymbol{\theta}_c} f(\mathbf{A}', \boldsymbol{\Theta}')\|  \\
    &\quad \leq \sum_{i=1}^{N} \Big[|a_{ic} - a_{ic}'| \cdot \| h_i(\boldsymbol{\theta}_c) J_i(\boldsymbol{\theta}_c) \|  \\
    &\quad \quad + a_{ic}' \| h_i(\boldsymbol{\theta}_c) J_i(\boldsymbol{\theta}_c) - h_i(\boldsymbol{\theta}_c') J_i(\boldsymbol{\theta}_c') \| \Big]  \\
    & \quad \leq \sum_{i=1}^{N} \Big[ M_J |a_{ic} - a_{ic}'| \\ 
    &  \quad  \quad + a_{ic}' (L_h M_J + M_g L_J) \|\boldsymbol{\theta}_c - \boldsymbol{\theta}_c'\| \Big].
\end{aligned}
\end{equation}  

The results in~\eqref{eq:grad_diff_ell},~\eqref{eq:grad_diff_D} and~\eqref{eq:grad_diff1} ensure that all requirements for~\eqref{eq:a1} in Assumption~\ref{as:smooth} are fulfilled. Thus, $f(\mathbf{A}, \boldsymbol{\Theta})$ is $L$-smooth.
\end{proof}

\subsection{Proof of Proposition~\ref{prop:bcd-convergence}}
We consider optimization problem~\eqref{eq:relaxed_opt}, with \(f\) and \(r_\mathbf{A}\) from~\eqref{eq:f}-\eqref{eq:r}, and with the block variables
\(\mathbf A\) 
and \(\boldsymbol\Theta\). 
We assume the Assumptions~\ref{as:smooth}--\ref{as:boundedness} and run updates for both blocks according to the framework in Section~\ref{sec:bcd}. We now present lemmas that lead towards the main convergence theorem.

\begin{lem}(\hspace{-5pt}~\cite[Equation $(2.8)$]{xu2013block})
\label{lem:one-step-decrease}
For each \(k\) the proximal \(\mathbf \Theta\)-update~\eqref{eq:block1} yields
\begin{equation}\label{eq:Theta-decrease}
F(\mathbf A^{k},\boldsymbol\Theta^{k})
- F(\mathbf A^{k},\boldsymbol\Theta^{k+1})
\;\ge\; \frac{L_{\boldsymbol \Theta}^k}{2}\,\|\boldsymbol\Theta^{k+1}-\boldsymbol\Theta^{k}\|_F^2.
\end{equation}

Similarly, the \(\boldsymbol A\)-update~\eqref{eq:block2} gives
\begin{equation}\label{eq:A-decrease}
F(\mathbf A^{k},\boldsymbol\Theta^{k+1})
- F(\mathbf A^{k+1},\boldsymbol\Theta^{k+1})
\;\ge\; \frac{L_{\mathbf A}^k}{2}\,\|\mathbf A^{k+1}-\mathbf A^{k}\|_F^2.
\end{equation}



Hence, each single-block step decreases the objective by at least a quadratic
term proportional to the step norm squared.
\end{lem}

\begin{proof}
We give the standard argument for proximal updates. 
By optimality of \(\boldsymbol{\Theta}^{k+1}\) in~\eqref{eq:block1} we have for any \(\boldsymbol{\Theta}\)
\begin{align*}
F(\mathbf A^{k},\boldsymbol\Theta^{k+1}) 
+ \frac{L_{\boldsymbol{\Theta}}^k}{2}\|\boldsymbol{\Theta}^{k+1}-\boldsymbol{\Theta}^{k}\|_F^2
&\le \\
F(\mathbf A^k,\boldsymbol\Theta) 
+ \frac{L_{\boldsymbol{\Theta}}^k}{2}\|\boldsymbol{\Theta}-\boldsymbol{\Theta}^{k}\|_F^2&.
\end{align*}
Take \(\boldsymbol\Theta=\boldsymbol\Theta^{k}\) and rearrange:
\begin{align*}
F(\mathbf A^{k},\boldsymbol\Theta^{k}) 
- F(\mathbf A^{k},\boldsymbol\Theta^{k+1})
\ge  \frac{L_{\boldsymbol{\Theta}}^k}{2}\|\boldsymbol{\Theta}^{k+1}-\boldsymbol{\Theta}^{k}\|_F^2&.
\end{align*}
Since \(F(\mathbf A,\boldsymbol\Theta)=f(\mathbf A,\boldsymbol\Theta)+r_\mathbf A(\mathbf A)\)
and only \(\mathbf A\) changed in this step,~\eqref{eq:A-decrease} follows.
\end{proof}

\begin{lem}(\hspace{-3.5pt}~\cite[Lemma $2.2$]{xu2013block})
\label{lem:telescoping}
Under the parameter bounds on \(L_{\mathbf A}^k,L_{\boldsymbol \Theta}^k\) in Assumption~\ref{as:block-updates} we have monotone decrease:
\begin{align*}
F(\mathbf A^{k+1},\boldsymbol\Theta^{k+1}) \le F(\mathbf A^{k},\boldsymbol\Theta^{k}) &\\- \gamma\Big(\|\mathbf A^{k+1}-\mathbf A^{k}\|_F^2
+ \|\boldsymbol\Theta^{k+1}-\boldsymbol\Theta^{k}\|_F^2\Big)&,
\end{align*}
for \(\gamma = \tfrac12\min\{\underline L_{\mathbf A},\underline L_{\boldsymbol \Theta}\}>0\). Consequently
the series \(\sum_{k=0}^\infty \big(\|\mathbf A^{k+1}-\mathbf A^{k}\|_F^2
+ \|\boldsymbol\Theta^{k+1}-\boldsymbol\Theta^{k}\|_F^2\big)\) is finite provided
\(F\) is bounded below (Assumption~\ref{as:boundedness}).
\end{lem}

\begin{proof}
Sum~\eqref{eq:Theta-decrease} and~\eqref{eq:A-decrease} to get
\begin{align*}
    F(\mathbf A^{k},\boldsymbol\Theta^{k})-F(\mathbf A^{k+1},\boldsymbol\Theta^{k+1}) & \ge \\ \frac{L_{\mathbf A}^k}{2}\|\mathbf A^{k+1}-\mathbf A^{k}\|_F^2
+ \frac{L_{\boldsymbol \Theta}^k}{2}\|\boldsymbol\Theta^{k+1}-\boldsymbol\Theta^{k}\|_F^2&.
\end{align*}

Using the uniform lower bounds $\underline L_{(\cdot)}$ yields the claimed inequality with
\(\gamma\) as above. Summing the inequalities over $k$ results in a telescoping sum of the objective, and, because $F$ is bounded below (Assumption~\ref{as:boundedness}), the sum of the right-hand terms is finite. 
\end{proof}

\begin{corollary}\label{cor:vanish}
From Lemma~\ref{lem:telescoping} we obtain
\[
\|\mathbf A^{k+1}-\mathbf A^{k}\|_F \to 0,\qquad
\|\boldsymbol\Theta^{k+1}-\boldsymbol\Theta^{k}\|_F \to 0,\quad\text{as }k\to\infty.
\]
\end{corollary}
\begin{proof}
Square-summability implies the terms tend to zero.
\end{proof}

\begin{lem}(\hspace{-3pt}~\cite[Theorem $2.3$]{xu2013block})
\label{lem:subsequence}
Suppose the iterates \(\{\mathbf A^{k},\boldsymbol\Theta^{k}\}\) are contained in a bounded set (Assumption~\ref{as:boundedness}). Then the sequence has at least one accumulation point
\(\bar x=(\bar{\mathbf A},\bar{\boldsymbol\Theta})\). Any such accumulation point is a
block-coordinate-wise stationary Nash point: i.e., it satisfies the first-order optimality condition of each block.
\end{lem}

\begin{proof}
Boundedness gives existence of a convergent subsequence \(x^{k_j}\to\bar x\).
By optimality of each proximal subproblem we have the proximal optimality inclusions
\begin{align*}
0 &= \nabla_{\boldsymbol\Theta} f(\mathbf A^{k},\boldsymbol\Theta^{k+1})
+ L_{\boldsymbol \Theta}^k(\boldsymbol\Theta^{k+1}-\boldsymbol\Theta^{k}), \\
0 &= \nabla_{\mathbf A} f(\mathbf A^{k+1},\boldsymbol\Theta^{k+1})
+ L_{\mathbf A}^k(\mathbf A^{k+1}-\mathbf A^{k}).
\end{align*}
Take limits along the subsequence \(k_j\). Using continuity of \(\nabla f\) and the fact
that the step differences vanish (Corollary~\ref{cor:vanish}), the inclusion limits to
\begin{align*}
0 &= \nabla_{\boldsymbol\Theta} f(\bar{\mathbf A},\bar{\boldsymbol\Theta}),\\
0 &= \nabla_{\mathbf A} f(\bar{\mathbf A},\bar{\boldsymbol\Theta}).
\end{align*}
These are exactly the block first-order optimality conditions (Nash conditions). 
\end{proof}

\begin{theo}
\label{thm:global}
Under the previous Assumptions~\ref{as:smooth}--\ref{as:boundedness}, the sequence \(\{\mathbf A^{k},\boldsymbol\Theta^{k}\}\) converges to a single limit point \(\bar x=(\bar{\mathbf A},\bar{\boldsymbol\Theta})\), which is a block-stationary
point. The convergence rate depends on the KL exponent at \(\bar x\).
\end{theo}

\begin{proof}
From Lemma~\ref{lem:telescoping} we have sufficient decrease and square-summability of step norms. The KL property provides a de-singularizing function \(\varphi\) so that near a limit point, the objective residual controls the norm of a suitable gradient. Combining the sufficient-decrease and square-summability results with the KL inequality 
yields that the sum of the norms of increments is finite:
\[
\sum_{k=0}^{\infty}\|x^{k+1}-x^k\|_F < +\infty,
\]
which implies Cauchy property and hence full-sequence convergence to a single limit point. The rates follow from the exponent in the KL inequality~\cite{xu2013block}. 
\end{proof}



\renewcommand{\bibname}{References}
\printbibliography

\end{document}